\documentclass{amsart}
\usepackage[T1]{fontenc}
\usepackage{amsmath,a4wide}
\usepackage{amssymb}
\usepackage{amscd}
\usepackage{epsfig}
\usepackage[T1]{fontenc}
\usepackage{epsfig}
\usepackage{amsmath}
\usepackage{mathabx}
\usepackage{graphicx}
\usepackage{subfigure}
\usepackage{mathrsfs}
\usepackage{epstopdf}
\usepackage{color}
\usepackage{booktabs}
\usepackage{ifpdf}
\usepackage{cite}
\usepackage{amsthm}
\usepackage{amsmath}
\usepackage{mathrsfs}
\usepackage{dsfont}
\usepackage{braket}
\definecolor{myblue}{rgb}{0.2,0,0.9}
\definecolor{blue-violet}{rgb}{0.54, 0.17, 0.89}
\usepackage{enumitem}
\setlist[enumerate]{label={\upshape(\roman*)}}
\usepackage[linesnumbered,ruled,vlined]{algorithm2e}
\SetKwInput{KwInput}{Input}
\SetKwInput{KwOutput}{Output}

\usepackage{pgfplots}
\usepgfplotslibrary{groupplots}
\pgfplotsset{compat=1.12}

\usepackage{tikz}
\usetikzlibrary{arrows,intersections}
\usepackage{boldline}
\usepackage{multirow}
\usepackage[margin=1.32in]{geometry}

\usepackage{varioref}
\usepackage{xr-hyper}
\usepackage{hyperref}
\hypersetup{
colorlinks,linkcolor=blue,citecolor=blue,filecolor=red,urlcolor=blue}

\def\e{\epsilon}

\definecolor{myblue}{rgb}{0.2,0,0.9}
\definecolor{blue_violet}{rgb}{0.54, 0.17, 0.89}
\definecolor{darkgreen}{rgb}{0,0.35,0}

\allowdisplaybreaks

\newcommand{\compconj}[1]{%
\overline{#1}%
}

\makeatletter
\DeclareRobustCommand*\cal{\@fontswitch\relax\mathcal}
\makeatother

\makeatletter
\newcommand{\labeltext}[3][]{%
\@bsphack%
\csname phantomsection\endcsname
\def\tst{#1}%
\def\labelmarkup{\emph}
\def\refmarkup{}%
\ifx\tst\empty\def\@currentlabel{\refmarkup{#2}}{\label{#3}}%
\else\def\@currentlabel{\refmarkup{#1}}{\label{#3}}\fi%
\@esphack%
\labelmarkup{#2}
}
\makeatother

\newtheorem{theorem}{Theorem}[section]
\newtheorem{proposition}[theorem]{Proposition}

\newtheorem{lemma}[theorem]{Lemma}

\numberwithin{equation}{section}

\theoremstyle{definition}
\newtheorem{remark}[theorem]{Remark}

\usepackage{xparse}
\makeatletter
\RenewDocumentCommand{\title}{om}{%
\IfNoValueTF{#1}
{\gdef\shorttitle{}}
{\gdef\shorttitle{#1}}%
\gdef\@title{#2}%
}
\makeatother

\makeatletter
\def\@tocline#1#2#3#4#5#6#7{\relax
\ifnum #1>\c@tocdepth 
\else
\par \addpenalty\@secpenalty\addvspace{#2}%
\begingroup \hyphenpenalty\@M
\@ifempty{#4}{%
	\@tempdima\csname r@tocindent\number#1\endcsname\relax
}{%
	\@tempdima#4\relax
}%
\parindent\z@ \leftskip#3\relax \advance\leftskip\@tempdima\relax
\rightskip\@pnumwidth plus4em \parfillskip-\@pnumwidth
#5\leavevmode\hskip-\@tempdima
\ifcase #1
\or\or \hskip 2em \or \hskip 2em \else \hskip 3em \fi%
#6\nobreak\relax
\hfill\hbox to\@pnumwidth{\@tocpagenum{#7}}\par
\nobreak
\endgroup
\fi}
\makeatother

\title{Approximation rates of quantum neural networks\\for periodic functions via Jackson's inequality}

\author{Ariel Neufeld}
\address{Nanyang Technological University, Singapore, ariel.neufeld@ntu.edu.sg}

\author{Philipp Schmocker}
\address{ETH Zurich, Switzerland, philipp.schmocker@math.ethz.ch}

\author{Viet Khoa Tran}
\address{Nanyang Technological University, Singapore, vietkhoa001@e.ntu.edu.sg}

\begin{document}

\begin{abstract}
	Quantum neural networks (QNNs) are an analog of classical neural networks in the world of quantum computing, which are represented by a unitary matrix with trainable parameters. Inspired by the universal approximation property of classical neural networks, ensuring that every continuous function can be arbitrarily well approximated uniformly on a compact set of a Euclidean space, some recent works have established analogous results for QNNs, ranging from single-qubit to multi-qubit QNNs, and even hybrid classical-quantum models. In this paper, we study the approximation capabilities of QNNs for \textit{periodic} functions with respect to the supremum norm. We use the Jackson inequality to approximate a given function by implementing its approximating trigonometric polynomial via a suitable QNN. In particular, we see that by restricting to the class of periodic functions, one can achieve a quadratic reduction of the number of parameters, producing better approximation results than in the literature. Moreover, the smoother the function, the fewer parameters are needed to construct a QNN to approximate the function. 
\end{abstract}

\maketitle

\section{Introduction}

Thanks to the development of quantum computing in recent years, quantum machine learning has emerged as a novel and interesting field (see, e.g., \cite{biamonte17,jacquier2022}). The goal of quantum machine learning is to enhance traditional machine learning tasks through the use of qubits and quantum circuit design. Hence, quantum neural networks (QNNs) arise naturally as an analog to traditional neural networks, with the potential to achieve quantum advantage in deep learning tasks.

In recent years, QNNs have gained significant interest among applied mathematicians and quantum computer scientists. One particular problem is whether QNNs can yield the universal approximation properties (UAP) which classical neural networks have (see, e.g., \cite{Cybenko1989,hornik89,pinkus99,neufeld24}). 
In \cite{PhysRevA.104.012405}, 
the authors construct a single-qubit quantum neural network to achieve universal approximation properties for continuous complex functions. A different quantum circuit is introduced in \cite{Schuld2021}, which shows that truncated Fourier series can be implemented as the expectation value of some observable with respect to a state prepared by a quantum circuit. The same idea is used in \cite{yu22} to prove the universal approximation theorem of single-qubit quantum neural networks for continuous univariate and square-integrable functions. The authors of \cite{goto2021} have proved the UAP of hybrid classical-quantum neural networks that encode the classical input into a quantum Hilbert space.
In \cite{Aftab2024Approximating}, the authors prove that QNNs can arbitrarily well approximate $d$-dimensional Korobov functions, and obtain worst-case complexity bounds on the  quantum circuit depth and number of qubits.
 In \cite{gonon23}, the authors consider continuous functions with integrable Fourier transforms and are able to obtain the UAP with respect to both the $L^2$ and $L^{\infty}$ norms as a consequence of obtaining the error bounds when approximating such functions by a quantum neural network. Building on the feedforward QNN introduced in \cite{gonon23}, the work in \cite{gonon2025} constructs a quantum analog of recurrent neural network, and proves the UAP of recurrent QNNs. Finally, in \cite{perezsalinas2025}, the authors obtain the UAP of quantum neural networks without pre-processing the input, while asymptotically improving the number of qubits and parameters compared to \cite{yu25}. This demonstrates the universal approximation prowess of QNNs for certain classes of functions and thus, paves the way for the application of QNNs in real-world machine learning tasks.

Another important problem related to the universal approximation property of QNNs is their approximation rate, which expresses the approximation error in terms of the QNN design, such as its depth, width, or the number of qubits used. In \cite{gonon23}, the authors require $\mathcal{O}(\epsilon^{-2})$ parameters and $\mathcal{O}(\log_2(\epsilon^{-1}))$ qubits to achieve an $L^2$-approximation bound of $\epsilon > 0$ when approximating functions with integrable Fourier transform, and for the $L^{\infty}$-error when approximating functions with additional integrable properties of the Fourier transform. 

In this paper, we turn our focus to the class of \textit{periodic} functions. Periodic functions model many repeating phenomena in real life, such as the motion of planets, the changing of seasons, sound and light waves, and even the flow of alternating current in electricity. Their repeating nature makes it possible to predict future behavior once the period is known, which is essential in science, engineering, and everyday life (see, e.g., \cite{oppenheim97,zygmund03,fourier09,oldham2010}). By using periodic functions, complex systems can be analyzed more simply and effectively, especially through tools like Fourier analysis that break down signals into basic waves (see, e.g., \cite{stein03,folland92,körner2022}).

At the same time, the ability of single-qubit QNNs to represent any trigonometric polynomials has been studied in \cite{ylwang23}. This initiates a new pathway to approximate continuous periodic functions, as their truncated Fourier series can be approximated by some single-qubit QNN of fixed length. Moreover, the approximation of a periodic function by a trigonometric polynomial can be controlled with the Jackson inequality (see, e.g., \cite{jackson11,lorentz1966,cheney82}). 

In our work, we will focus on approximating periodic functions. In the case of such univariate functions, for a prescribed error tolerance level $\e > 0$, we require $\mathcal{O}(\epsilon^{-1})$ parameters and only one qubit. In the case of $d$-variate functions which are $(K+1)$-times continuously differentiable with $K \in \mathbb{N}_0$ satisfying $K \geq d$, we can construct a multi-qubit QNN using $\mathcal{O}(\epsilon^{-1})$ parameters and $\mathcal{O}(\log_2(\epsilon^{-1}))$ qubits to approximate such functions. Compared to previous work on approximation rates of QNN in \cite{gonon23}, we achieve a quadratic reduction in the number of parameters when approximating periodic and smooth functions without having to impose any integrability condition on their Fourier transforms. 

\subsection{Notation}

As usual, $\mathbb{N} := \lbrace 1, 2, 3, \dots \rbrace$ and $\mathbb{N}_0 := \mathbb{N} \cup \lbrace 0 \rbrace$ denote the sets of natural numbers, $\mathbb{Z}$ represents the set of integers, whereas $\mathbb{R}$ and $\mathbb{C}$ are the sets of real and complex numbers, respectively, with $\mathbf{i}$ denoting the imaginary unit satisfying $\mathbf{i}^2 := -1$. The complex conjugate of $c := a + b\mathbf{i} \in \mathbb{C}$ is defined as $\compconj{c}= a - b\mathbf{i} \in \mathbb{C}$, where $a,b \in \mathbb{R}$. Moreover, we use the abbreviation $\sum_{n=a-1/2}^{b-1/2} c_n := \sum_{n=a}^b c_{n-1/2}$ for $a,b \in \mathbb{Z}$ with $a \leq b$ and $(c_t)_{t \in \mathbb{R}} \subseteq \mathbb{R}$. In addition, for every fixed $d \in \mathbb{N}$, we denote a multi-index $\mathbf{n} := (n_1,\dots,n_d) \in \mathbb{Z}^d$ in bold letters, where we define $\Vert \mathbf{n} \Vert_1 := \vert n_1 \vert + \dots + \vert n_d \vert$ for $\mathbf{n} := (n_1,\dots,n_d) \in \mathbb{N}_0^d$. For $\mathbf{n}, \mathbf{m}\in \mathbb{Z}^d$, we say that $\mathbf{n} \leq \mathbf{m}$  if $n_i \leq m_i$ for all $i = 1,\dots,d$. Furthermore, we denote by $\mathbb{R}^d$ (resp.~$\mathbb{C}^d$) the $d$-dimensional vector space over $\mathbb{R}$ (resp.~$\mathbb{C}$) equipped with the norm $\Vert \mathbf{x} \Vert = \big(\sum_{i=1}^d \vert x_i \vert^2 \big)^{1/2}$ for $\mathbf{x} := (x_1,\dots,x_d)^\top \in \mathbb{R}^d$ (resp.~$x := (x_1,\dots,x_d)^\top \in \mathbb{C}^d$). Moreover, $\mathbb{R}^{d \times d}$ (resp.~$\mathbb{C}^{d \times d}$) represents the space of matrices $A := (A_{i,j})_{i,j=1,\dots,d} \in \mathbb{R}^{d \times d}$ (resp. $A := (A_{i,j})_{i,j=1,\dots,d} \in \mathbb{C}^{d \times d}$). Hereby, we denote by $I_d \in \mathbb{C}^{d \times d}$ the identity matrix, use the notation $U^\dagger \in \mathbb{C}^{d \times d}$ for the Hermitian adjoint of a matrix $U \in \mathbb{C}^{d \times d}$, and recall that a matrix $U \in \mathbb{C}^{d \times d}$ is unitary if $U^\dagger U = U U^\dagger = I_d$. We further denote the group of unitary matrices $U \in \mathbb{C}^{d \times d}$ by $\mathcal{U}(d)$. For any $n \in \mathbb{N}$ and matrix $M \in \mathbb{C}^{d\times d}$, we write $M^{\otimes n} \in \mathbb{C}^{d^n \times d^n}$ to denote the $n$-fold Kronecker product of $M$.

In addition, for every $d \in \mathbb{N}$, we denote by $C(\mathbb{R}^d)$ the vector space of continuous functions $f: \mathbb{R}^d \rightarrow \mathbb{R}$ equipped with the supremum norm $\Vert f \Vert_\infty := \sup_{x \in \mathbb{R}^d} \vert f(x) \vert$. Then, for every $f \in C(\mathbb{R}^d)$, we denote its \emph{(global) modulus of continuity} by
\begin{equation*}
	[0,\infty) \ni \delta \quad \mapsto \quad \omega_f(\delta) := \sup_{x,t \in \mathbb{R} \atop \vert t \vert < \delta} \vert f(x+t) - f(x) \vert \in [0,\infty].
\end{equation*}
Furthermore, we denote by $C_{2\pi}(\mathbb{R}^d) \subset C(\mathbb{R}^d)$ the vector subspace of continuous and \emph{$2\pi$-periodic functions} $f: \mathbb{R}^d \rightarrow \mathbb{R}$. Hereby, a function $f: \mathbb{R}^d \rightarrow \mathbb{R}$ is called \textit{$2\pi$-periodic} if for every $j = 1,\dots,d$ and $(x_1,\dots,x_d) \in \mathbb{R}^d$ we have $f(x_1,\dots,x_{j-1},x_j+2\pi,x_{j+1},\dots,x_d) = f(x_1,\dots,x_{j-1},x_j,x_{j+1},\dots,x_d)$. In this case, we denote the Fourier coefficients of every $f \in C_{2\pi}(\mathbb{R}^d)$ by $(\widehat{f}(\mathbf{n}))_{\mathbf{n} \in \mathbb{Z}^d}$ that are for every $\mathbf{n} \in \mathbb{Z}^d$ defined by
\begin{equation*}
	\widehat{f}(\mathbf{n}) := \frac{1}{(2\pi)^d} \int_{[-\pi,\pi]^d} e^{-\mathbf{i} \mathbf{n}^\top \mathbf{x}} f(\mathbf{x}) d\mathbf{x}.
\end{equation*}
Then, by applying \cite[Theorem~2.1]{folland92} to every variable $x_j$, $j = 1,\dots,d$, it follows for every $\mathbf{x} \in \mathbb{R}^d$ that $f(\mathbf{x}) = \sum_{\mathbf{n} \in \mathbb{Z}^d} \widehat{f}(\mathbf{n}) e^{\mathbf{i} \mathbf{n}^\top \mathbf{x}}$.

Moreover, we use the Landau $\mathcal{O}$-notation, i.e., $a_N = \mathcal{O}(b_N)$ if and only if $\limsup_{N \rightarrow \infty} \big\vert \frac{a_N}{b_N} \big\vert < \infty$.

\subsection{Outline}

The remainder of the paper is structured as follows. Section 2 provides a brief introduction to quantum computing, including the concepts of quantum bits, quantum gates, quantum circuits, and measurements. The specific construction of the QNNs used in this paper will also be discussed in this section. Section 3 presents the main results in two separate theorems, one for the case of univariate functions and one for the case of multivariate functions. In Section 4, we will discuss the numerical results obtained when approximating continuous periodic functions using QNNs. The complete proofs of our main results will be shown in Sections 5 and~6.

\section{Construction of quantum neural networks}

In this section, we provide a brief introduction to quantum computing following the textbooks~\cite{NielsenChuang10, jacquier2022}, and the construction of quantum neural networks (QNNs) which will be used in this paper.

\subsection{Bra-Ket notation}
We start with a finite-dimensional complex Hilbert space $\mathcal{H}$ and its dual space $\mathcal{H}^\ast$. A vector $\psi \in \mathcal{H}$, often called a \textit{state}, is denoted by the \textit{ket} notation $\ket \psi$, and the linear functional $(\psi \mapsto\langle\varphi,\psi\rangle )\in \mathcal{H}^\ast$ where $\varphi \in \mathcal{H}$ is denoted by the \textit{bra} notation $\bra \varphi$. Moreover, the action of $\bra \varphi$ on $\ket \psi$ is denoted by the \textit{bra-ket} notation $\langle \psi| \varphi \rangle := \langle \psi,\varphi\rangle \in \mathbb{C}$. The action of a linear operator $\mathcal{A}: \mathcal{H} \to \mathcal{H}$ on a vector $\psi \in \mathcal{H}$ is represented by $\mathcal{A}\ket\psi$. Similarly, the action of $\mathcal{A}$ on $\bra\varphi \in \mathcal{H}^\ast$, represented by $\bra\varphi\mathcal{A}$, is a linear functional that satisfies $(\bra\varphi\mathcal{A})\ket\psi = \bra\varphi(\mathcal{A}\ket\psi) := \langle\varphi,\mathcal{A}\psi\rangle$. We often write $\bra\varphi\mathcal{A}\ket\psi$ for brevity. The expectation value of an operator $\mathcal{A}$ with respect to a normalized state $\psi \in \mathcal{H}$, i.e., $\langle \psi | \psi\rangle = 1$, is $\langle \psi | \mathcal{A}| \psi\rangle$. Meanwhile, $\ket \psi \bra \varphi$ denotes the outer product of $\psi, \varphi \in \mathcal{H}$, which is a linear operator that maps $\mathcal{H}$ to itself such that $(\ket \psi \bra \varphi) \ket \xi = (\langle\varphi|\xi\rangle)\ket\psi$ for all $\xi \in \mathcal{H}$.

In quantum computing, we first consider the Hilbert space $\mathcal{H} \simeq \mathbb{C}^2$ and then the $n$-fold tensor product Hilbert space $\mathcal{H}^{\otimes n} := \mathcal{H}\otimes\dots\otimes\mathcal{H} \simeq \mathbb{C}^{2^n}$. This is sometimes called the \textit{state space}, and we denote by $\ket \psi _n$ the state $\psi \in \mathbb{C}^{2^n}$. The computational basis of the state space can then be constructed as $\mathcal{B}_n := \lbrace \Ket{\mathbf{j}}_n: \mathbf{j} := (j_1,\dots,j_n) \in \lbrace 0,1 \rbrace^n \rbrace$, with $\Ket{\mathbf{j}}_n := \Ket{j_1} \otimes \cdots \otimes \Ket{j_n} := \Ket{j_1} \cdots \Ket{j_n}$ for $\mathbf{j} := (j_1,\dots,j_n) \in \lbrace 0,1 \rbrace^n$. When the number of qubits $n$ is fixed, and by a slight abuse of notation, for any integer $k$ between $0$ and $2^n-1$, we write $\ket{k}$ to denote the computational basis state $\Ket{k_1} \cdots \Ket{k_n}$, where $(k_1 \dots k_n) \in \{0,1\}^n$ is the binary representation of $k$.

\subsection{Quantum bit}
In contrast to classical computing using classical bits, quantum computing relies on quantum bits, also called qubits, to store and process information. A single qubit is a normalized vector in the state space $\mathbb{C}^2$. Therefore, the set of all possible states is the set of $\ket{\psi}$ of the form
$$\ket{\psi} = \psi_0\ket{0} + \psi_1\ket{1} = \begin{bmatrix} \psi_0 \\ \psi_1 \end{bmatrix},$$
where $\psi_0, \psi_1 \in \mathbb{C}$ are such that $|\psi_0|^2 + |\psi_1|^2 = 1$. Here, $|\psi_0|^2$ and $|\psi_1|^2$ represent the probabilities of the qubit taking the value of $\ket{0}:=\begin{bmatrix} 1 \\ 0 \end{bmatrix}$ and $\ket{1}:=\begin{bmatrix} 0 \\ 1 \end{bmatrix}$, respectively.

Every $n$-qubit state $\Ket{\psi}_n$ can then be written as a linear combination of the computational basis, i.e., there exist some coefficients $(\psi_\mathbf{j})_{\mathbf{j} \in \lbrace 0,1 \rbrace^n} \subseteq \mathbb{C}$ with $\sum_{\mathbf{j} \in \lbrace 0,1 \rbrace^n} \vert \psi_\mathbf{j} \vert^2 = 1$ such that
\begin{equation*}
	\Ket{\psi}_n := \sum_{\mathbf{j} \in \lbrace 0,1 \rbrace^n} \psi_\mathbf{j} \Ket{\mathbf{j}}_n.
\end{equation*}
Here, the coefficients $(\psi_\mathbf{j})_{\mathbf{j} \in \lbrace 0,1 \rbrace^n}$ are called probability amplitudes (or simply amplitudes) since for every $\mathbf{j} \in \lbrace 0,1 \rbrace^n$, $\vert \psi_\mathbf{j} \vert^2 = |\langle \mathbf{j}, \psi \rangle_n|^2$ is the probability of $\ket{\psi}_n$ collapsing to the state $\ket {\mathbf{j}}_n$ after measurement.

\subsection{Quantum gate}
\label{sectionQgate}
In quantum computing, the evolution of qubit states is unitary. A quantum gate acting on $n$ qubits can, therefore, be represented as a unitary matrix of size $2^n$. We denote the state evolution of an $n$-qubit system $\ket{\psi} \to \ket{\varphi} \in \mathbb{C}^{2^n}$ via a quantum gate $U \in \mathcal{U}(2^n)$ by $U\ket{\psi} = \ket{\varphi}$. 

By elementary gates, we refer to the set in $\mathcal{U}(2) \cup \mathcal{U}(4)$ defined by
\begin{equation*}
	\mathbb{G} := \big\lbrace X,Y,Z,H,R_X(\theta), R_Y(\theta), R_Z(\theta), P(\vartheta): \theta \in (0,4\pi), \, \vartheta \in (0,2\pi) \big\rbrace \cup \big\lbrace \text{CNOT}, \text{SWAP} \big\rbrace,
\end{equation*}
where
\[
H := \frac{1}{\sqrt{2}}\begin{bmatrix}
	1 & 1 \\
	1 & -1
\end{bmatrix},
\quad
X :=
\begin{bmatrix}
	0 & 1 \\
	1 & 0
\end{bmatrix},
\quad
Y :=
\begin{bmatrix}
	0 & -\mathbf{i} \\
	\mathbf{i} & 0
\end{bmatrix},
\quad
Z := 
\begin{bmatrix}
	1 & 0 \\
	0 & -1
\end{bmatrix}
,\]
\[
R_X(\theta) := 
\begin{bmatrix}
	\cos (\frac{\theta}{2}) & -\mathbf{i} \sin (\frac{\theta}{2}) \\
	-\mathbf{i} \sin (\frac{\theta}{2}) & \cos (\frac{\theta}{2})
\end{bmatrix},
\quad
R_Y(\theta) :=
\begin{bmatrix}
	\cos (\frac{\theta}{2}) & -\sin (\frac{\theta}{2}) \\
	\sin (\frac{\theta}{2}) & \cos (\frac{\theta}{2})
\end{bmatrix},
\quad
R_Z(\theta) :=
\begin{bmatrix}
	e^{-\mathbf{i} \frac{\theta}{2}} & 0 \\
	0 & e^{\mathbf{i} \frac{\theta}{2}}
\end{bmatrix}
,\]
with $\mathbf{i} := \sqrt{-1}$ being the imaginary unit and $\theta \in \mathbb{R}$ being a parameter. These gates form a universal quantum gate set in the sense of the Solovay-Kitaev theorem \cite{dawsonnielsen06}. Here, $P(\vartheta)$ with $\vartheta \in (0,2\pi)$, CNOT, SWAP are other standard gates in quantum computing and are not introduced in this paper. Detailed discussions of these gates can be found in \cite{NielsenChuang10, jacquier2022}.

A quantum circuit $U$ of length $\mathscr{L} \in \mathbb{N}$ and acting on $n \in \mathbb{N}$ qubits is a unitary operator such that 
\begin{equation*}
	U = \prod\limits_{j=1}^{\mathscr{L} } (G_{j,n_1} \otimes G_{j,n_2} \otimes \dots \otimes G_{j, n_j}) \quad \in \mathcal{U}(2^n),
\end{equation*}
where $G_{1,n_1}, \dots G_{\mathscr{L} ,n_\mathscr{L} } \in \mathbb{G} \cup I_2$ are such that $(G_{j,n_1}\otimes \dots \otimes G_{j, n_j}) \in \mathcal{U}(2^n)$ for all $j = 1,\dots,\mathscr{L} $ (see also \cite{chenlineufeld23}). We regard any quantum circuit whose elementary gates are parameter-dependent as a quantum neural network (QNN).

\begin{remark}[\textbf{Ancilla qubits}]
	Ancilla qubits are qubits that are necessary for the storage of extra information in a quantum algorithm, even though they may not be needed for the final output. The number of ancilla qubits required can be used to measure the complexity of a quantum algorithm, and are counted together with other qubits in a multi-qubit quantum circuit.
\end{remark}

\subsection{Single-qubit QNN}
\label{sectionSingleQNN}

In our work, we consider the single-qubit quantum neural network (QNN) of the form

\begin{equation}
	\label{EqDefSingleCircuit}
	U^L_{\theta, \phi}(x):= R_Z(\varphi)R_Y(\theta_0)R_Z( \phi_0) \prod_{l=1}^{L} R_Z(x) R_Y(\theta_l)R_Z( \phi_l) \quad \in \mathcal{U}(2),
\end{equation}
which is introduced in \cite{yu22}. Here, $L$ is the depth of the QNN, $x \in \mathbb{R}$ is the input, and $\theta = (\theta_0, \dots, \theta_L) \in \mathbb{R}^{L+1}$ and $\phi = (\varphi, \phi_0, \dots, \phi_L) \in \mathbb{R}^{L+2}$ are the parameters of the quantum circuit. $R_Y(\cdot)$ and $R_Z(\cdot)$ are the rotation gates introduced in Section~\ref{sectionQgate}.

The output of the QNN defined above is the probability amplitude of the quantum state, after evolving under the QNN, collapsing to the state $\ket{0}$ when measured, that is
\begin{equation}
	\label{EqOutSingleCircuit}
	f^L_{\theta, \phi}(x):= \bra{0} U^L_{\theta, \phi}(x)\ket{0} .
\end{equation}

\subsection{Multi-qubit QNN}
\label{sectionMultiQNN}

For $d,n \in \mathbb{N}$ with $d \leq n$, an $n$-qubit quantum neural network (QNN) is represented as a unitary matrix $U_{\boldsymbol\theta, \boldsymbol\phi}(\mathbf{x}) \in \mathcal{U}(2^n)$, which depends on the input $\mathbf{x} \in \mathbb{R}^d$, in a parametric manner. We also let $q = n-d$ denote the number of ancilla qubits used in the circuit. We will now present the QNN used in our paper, which is based on the idea of \textit{linear combination of unitaries} (LCU) in \cite{yu25} and \cite{childs12}.

For every vector $\mathbf{L} = (L_1,\dots,L_d) \in \mathbb{N}^d$, we consider the set $ \{\mathbf{n} \in \mathbb{Z}^d, -\mathbf{L} \leq \mathbf{n} \leq \mathbf{L}\}$. We also define $\mathfrak{n} := |\{\mathbf{n} \in \mathbb{Z}^d, -\mathbf{L} \leq \mathbf{n} \leq \mathbf{L}\}| = \sum\limits_{j=1}^d(2L_j + 1)$. Assume on $\{\mathbf{n} \in \mathbb{Z}^d, -\mathbf{L} \leq \mathbf{n} \leq \mathbf{L}\}$ an arbitrary ordering of its elements $\{\mathbf{n} \in \mathbb{Z}^d, -\mathbf{L} \leq \mathbf{n} \leq \mathbf{L}\} = \{\mathbf{n}_0,\dots, \mathbf{n}_{\mathfrak{n}-1} \}$. Our QNN uses $q = \lceil \log_2(\mathfrak{n}) \rceil$ ancilla qubits. 

We first define, for each vector $\mathbf{n}_i := (n_{i,1}, \dots,n_{i,d}) \in \{\mathbf{n}_0,\dots, \mathbf{n}_{\mathfrak{n}-1} \} \subset \mathbb{Z}^d$, a unitary matrix
\begin{equation}
	U^{\mathbf{n}_i}_{\theta_{\mathbf{n}_i},\phi_{\mathbf{n}_i}}(\mathbf{x}) := \bigotimes_{j=1}^d U^{2|n_{i,j}|}_{\theta_{\mathbf{n}_i,j},\phi_{\mathbf{n}_i,j}}(x_j) \quad \in \mathcal{U}(2^d),
\end{equation}
where $\mathbf{x} \in \mathbb{R}^d$ is the input, $\theta_{\mathbf{n}_i} := (\theta_{\mathbf{n}_i,1},\dots,\theta_{\mathbf{n}_i,d}) \in \bigtimes_{j=1}^d \mathbb{R}^{2|n_{i,j}|+1}$ and $\phi_{\mathbf{n}_i} := (\phi_{\mathbf{n}_i,1},\dots,\phi_{\mathbf{n}_i,d}) \in \bigtimes_{j=1}^d \mathbb{R}^{2|n_{i,j}|+2}$ are the parameters, and $U^{2|n_{i,j}|}_{\theta_{\mathbf{n}_i,j},\phi_{\mathbf{n}_i,j}}(x_j)$ is the single-qubit QNN defined in Section~\ref{sectionSingleQNN}. 

We now define for every $\mathbf{x} \in \mathbb{R}^d$ the following unitary operator
\begin{equation}
	\label{EqMultiQubitOper}
	C^{\mathbf{L}}_{\boldsymbol\theta,\boldsymbol\phi}(\mathbf{x}) := \sum_{i = 0}^{\mathfrak{n}-1} \ket{i} \bra{i} \otimes U^{\mathbf{n}_i}_{\theta_{\mathbf{n}_i},\phi_{\mathbf{n}_i}}(\mathbf{x}) 
	+ \sum_{i=\mathfrak{n}}^{2^q-1} \ket{i} \bra{i} \otimes X^{\otimes d} \quad \in \mathcal{U}(2^{q+d}),
\end{equation}
where $X$ is the Pauli $X$ gate defined in Section~\ref{sectionQgate}. Finally, we introduce for every $\mathbf{x} \in \mathbb{R}^d$ the $(q+d)$-qubit operator
\begin{equation}
	\label{EqMultiCircuit}
	U^{\mathbf{L}}_{\boldsymbol\theta,\boldsymbol\phi}(\mathbf{x}) := \left( H^{\otimes q} \otimes I_2^{\otimes d} \right)^\dagger C^{\mathbf{L}}_{\boldsymbol\theta,\boldsymbol\phi}(\mathbf{x}) \left( H^{\otimes q} \otimes I_2^{\otimes d} \right) \quad \in \mathcal{U}(2^{q+d}),
\end{equation}
where $H$ is the Hadamard gate defined in Section~\ref{sectionQgate}. This is the multi-qubit QNN used in this paper. For every input $\mathbf{x} \in \mathbb{R}^d$, the output of the QNN is then defined as the amplitude of the quantum state, after evolving under the QNN, collapsing to the state $\ket{0}_{q+d}$ when measured, that is
\begin{equation}
	\label{EqMultiCircuitOut}
	f^{\mathbf{L}}_{\boldsymbol\theta,\boldsymbol\phi}(\mathbf{x}) := \bra{0}_{q+d} U^{\mathbf{L}}_{\boldsymbol\theta,\boldsymbol\phi}(\mathbf{x}) \ket{0}_{q+d}.
\end{equation}

\pagebreak

\section{Main results}

Our main results will be discussed in this section. Our first result, presented in Theorem~\ref{thrmJacksonK}, gives an upper bound for approximating a real, periodic, univariate function via a single-qubit quantum neural network (QNN). Our second result, presented in Theorem~\ref{thrmJacksonMultiK}, extends Theorem~\ref{thrmJacksonK} to approximating real, periodic, multivariate functions.

\subsection{Approximation rates for univariate functions}

We have the following result for the case of real, periodic univariate functions.

\begin{theorem}
	\label{thrmJacksonK}
	For $K \in \mathbb{N}_0$, let $f \in C_{2\pi}(\mathbb{R})$ be $K$-times continuously differentiable and define $c := (2^{K+1}-1)  \Vert f \Vert_\infty  > 0$. Then, there exists a constant $C_K > 0$ (independent of $f$) such that for every $N \in \mathbb{N}$ there exists a single-qubit QNN $U^{2L}_{\theta, \phi}$ as defined in \eqref{EqDefSingleCircuit} with depth $2L$, $L := \big\lceil \frac{K+3}{2} \big\rceil \big\lfloor \frac{N}{2} \big\rfloor$, and parameters $\theta \in \mathbb{R}^{2L+1}$ and $\phi \in \mathbb{R}^{2L+2}$ satisfying
	\begin{equation}
		\label{EqCorJacksonK}
		\sup_{x \in \mathbb{R}} \left\vert f(x) - c \cdot f^{2L}_{\theta,\phi}(x) \right\vert \leq \frac{C_K \omega_{f^{(K)}}\left( \tfrac{1}{N} \right)}{N^K}.
	\end{equation}
	Moreover, if $f \in C_{2\pi}(\mathbb{R})$ is $(K+1)$-times continuously differentiable, then the right-hand side of \eqref{EqCorJacksonK} can be upper bounded by $\frac{C_K \Vert f^{(K+1)} \Vert_\infty}{N^{K+1}}$, where $f^{(K+1)}$ denotes the $(K+1)^{\text{th}}$ derivative of~$f$.
\end{theorem}

\begin{remark}[\textbf{Complexity analysis}]
	We will discuss the complexity of our QNN in terms of the number of parameters and qubits used with respect to the prescribed error tolerance level $\epsilon > 0$. The approximation error in this case is the $L^{\infty}$-norm error and is defined as $\sup_{x \in \mathbb{R}} \left\vert f(x) - c \cdot f^{2L}_{\theta,\phi}(x)\right\vert$. From the second part of Theorem~\ref{thrmJacksonK}, when $K=0$, $f \in C_{2\pi}(\mathbb{R})$ is continuously differentiable and its first derivative is bounded, we have $\frac{C_0 \Vert f^{(1)} \Vert_\infty}{N} < \epsilon$ if and only if $N > \frac{C_0\Vert f^{(1)} \Vert_\infty}{\e}$. Therefore, we can choose $N \in \mathbb{N}$ such that $N = \mathcal{O}(\epsilon^{-1})$ and thus the number of parameters $\Theta$ is of order $\Theta = \mathcal{O}(L) = \mathcal{O}(N) = \mathcal{O}(\epsilon^{-1})$. When $K \in \mathbb{N}$ and $f$ is a $K+1$-times continuously differentiable function, we require fewer parameters to approximate $f$ with the approximation error level $\epsilon > 0$. In detail, by setting $k := K+1$ and assuming that the function has a bounded $k$-th derivative, we have $N^k = \mathcal{O}(\epsilon^{-1})$, which implies that $N = \mathcal{O}(\epsilon^{-1/k})$. Therefore, for fixed $k \in \mathbb{N}$, the number of parameters $\Theta$ required by the QNN is of order $\Theta = \mathcal{O}(L) = \mathcal{O}(Nk) = \mathcal{O}(\epsilon^{-1/k})$. Moreover, for all $k \in \mathbb{N}$, the QNN uses only one qubit.
\end{remark}

\begin{remark}[\textbf{General periodic functions}]
	Theorem~\ref{thrmJacksonK} can also be applied to periodic functions with other frequency than $2\pi$. That is, given a continuous $M$-periodic function $\widetilde{f}: \mathbb{R} \rightarrow \mathbb{R}$, with $M > 0$, we learn the $2\pi$-periodic function
	\begin{equation*}
		\mathbb{R} \ni x \quad \mapsto \quad f(x) := \widetilde{f}\left( \frac{M x}{2\pi} \right) \in \mathbb{R}.
	\end{equation*}
	The QNN's output $f^{2L}_{\theta,\phi}$ in Theorem~\ref{thrmJacksonK} can then be rescaled to obtain a function
	\begin{equation*}
		\mathbb{R} \ni x \quad \mapsto \quad \widetilde{f}^{2L}_{\theta,\phi}(x) := f^{2L}_{\theta,\phi}\left( \frac{2\pi x}{M} \right) \in \mathbb{R},
	\end{equation*} 
	which approximates the $M$-periodic function $\widetilde{f}$, with rate $\frac{C_K \omega_{f^{(K)}}(\frac{1}{N})}{N^K}  = (\frac{2\pi}{M})^K \frac{C_K \omega_{\widetilde{f}^{(K)}}(\frac{1}{N})}{N^K}$.
\end{remark}

\subsection{Approximation rates for multivariate functions}

We have the following result for the case of real, periodic multivariate functions. Given any two $d$-dimensional vectors $\mathbf{N} \in \mathbb{N}^d$ and $\mathbf{K} \in \mathbb{N}_0^d$, we introduce the vector 
\begin{equation}
	\label{eqL}
	\mathbf{L}_{\mathbf{N}, \mathbf{K}} := \left(\left\lceil \frac{K_1+3}{2} \right\rceil\left\lfloor \frac{N_1}{2} \right\rfloor, \dots, \left\lceil \frac{K_d+3}{2} \right\rceil\left\lfloor \frac{N_d}{2} \right\rfloor\right) \in \mathbb{N}_0^d. 
\end{equation}
\begin{theorem}
	\label{thrmJacksonMultiK}
	For $\mathbf{K} := (K_1,\dots,K_d) \in \mathbb{N}_0^d$, let $f \in C_{2\pi}(\mathbb{R}^d)$ have a $K_j$-th continuous partial derivative $\partial_j^{K_j} f := \frac{\partial^{K_j} f}{\partial x_j^{K_j}}: \mathbb{R}^d \rightarrow \mathbb{R}$, $j = 1,\dots,d$. Moreover, let $c := \Vert f \Vert_\infty \prod_{j=1}^d(2^{K_j + 1}-1) \geq 0$. Then, there exists a constant $C_{\mathbf{K}} > 0$ (independent of $f$) such that for every $\mathbf{N} := (N_1,\dots,N_d) \in \mathbb{N}^d$ there exists a $(q+d)$-qubit QNN $U^{\mathbf{L}_{\mathbf{N}, \mathbf{K}}}_{\boldsymbol\theta, \boldsymbol\phi}$ as defined in \eqref{EqMultiCircuit} with parameters $\boldsymbol\theta := (\theta_\mathbf{n})_{-\mathbf{L}_{\mathbf{N}, \mathbf{K}} \leq \mathbf{n} \leq \mathbf{L}_{\mathbf{N}, \mathbf{K}}} \in \bigtimes_{-\mathbf{L}_{\mathbf{N}, \mathbf{K}} \leq \mathbf{n} \leq \mathbf{L}_{\mathbf{N}, \mathbf{K}}} \bigtimes_{j=1}^d \mathbb{R}^{2 \vert n_j \vert+1}$ and $\boldsymbol\phi := (\phi_\mathbf{n})_{-\mathbf{L}_{\mathbf{N}, \mathbf{K}}\leq \mathbf{n} \leq \mathbf{L}_{\mathbf{N}, \mathbf{K}}} \in \bigtimes_{-\mathbf{L}_{\mathbf{N}, \mathbf{K}} \leq \mathbf{n} \leq \mathbf{L}_{\mathbf{N}, \mathbf{K}}} \bigtimes_{j=1}^d \mathbb{R}^{2\vert n_j \vert+2}$ satisfying
	\begin{equation}
		\label{EqCorJacksonMultiK}
		\sup_{x \in \mathbb{R}^d} \left\vert f(\mathbf{x}) - c \cdot 2^q \cdot f^{\mathbf{L}_{\mathbf{N}, \mathbf{K}}}_{\boldsymbol\theta,\boldsymbol\phi}(\mathbf{x}) \right\vert \leq C_\mathbf{K} \sum_{j=1}^d \frac{\omega_{\partial_j^{K_j} f}\big(\tfrac{1}{N_j}\big)}{N_j^{K_j}},
	\end{equation}
	where $\mathfrak{n} := |\{\mathbf{n} \in \mathbb{Z}^d, -\mathbf{L}_{\mathbf{N}, \mathbf{K}} \leq \mathbf{n} \leq\mathbf{L}_{\mathbf{N}, \mathbf{K}}\}|$ and $q := \lceil \log_2(\mathfrak{n}) \rceil$. Moreover, if $f \in C_{2\pi}(\mathbb{R}^d)$ has a ($K_j+1$)-th continuous partial derivative $\partial_j^{K_j+1} f: \mathbb{R}^d \rightarrow \mathbb{R}$, $j = 1,\dots,d$, then the right-hand side of \eqref{EqCorJacksonMultiK} can be upper bounded by $C_\mathbf{K} \sum_{j=1}^d \frac{\Vert \partial_j^{K_j+1} f \Vert_\infty}{N_j^{K_j+1}}$.
\end{theorem}

\begin{remark}[\textbf{Complexity analysis}]
	For simplicity, we set $N := N_1 = \dots=N_d$, $K:=K_1= \dots=K_d$ and $k:=K+1$. In the case where $f \in C_{2\pi}(\mathbb{R}^d)$ has all $k$-th continuous partial derivatives, we have that $N^k = \mathcal{O}(d\epsilon^{-1})$, and therefore $N = \mathcal{O}((d\epsilon^{-1})^{1/k})$, where $\e > 0$ is the prescribed error tolerance level. The number of parameters $\Theta$ in our QNN is of order
	\begin{equation}
		\label{upperboundw}
		\Theta = \mathcal{O}\left( \sum\limits_{-\mathbf{L} \leq \mathbf{n} \leq \mathbf{L}} \Vert \mathbf{n} \Vert_1 \right) = \mathcal{O}(d(Nk)\cdot\mathfrak{n}) = \mathcal{O}(d(Nk)(2Nk+1)^d).
	\end{equation}
	Substituting $N = \mathcal{O}((d\epsilon^{-1})^{1/k})$ into \eqref{upperboundw}, we have
	\begin{equation*}
		\Theta = \mathcal{O}(d\cdot3^d\cdot d^{(d+1)/k} \cdot k^{d+1} \cdot \epsilon^{-(d+1)/k}) = \mathcal{O}(3^d d^{(d+k+1)/k}k^{d+1}\cdot \epsilon^{-(d+1)/k}).
	\end{equation*} 
	Furthermore, the number of ancilla qubits used in this set-up is
	\begin{equation*}
		q = \mathcal{O}(\log_2(\mathfrak{n})) = \mathcal{O}\left(\frac{d}{k}\log_2(d\epsilon^{-1}) + d\log_2(k)\right) = \mathcal{O}\left(\frac{d}{k}\log_2(d) + d\log_2(k)+\frac{d}{k}\log_2(\epsilon^{-1})\right).
	\end{equation*}
	The total number of qubits is just $$d+q = \mathcal{O}\left(d+\frac{d}{k}\log_2(d) + d\log_2(k)+\frac{d}{k}\log_2(\epsilon^{-1})\right).$$
\end{remark}

\begin{remark}[\textbf{Comparison with previous works}]
	For fixed $d, k\in\mathbb{N}$, we compare the number of parameters as well as the number of qubits used to \cite{gonon23}. There, for any prescribed error level $\e>0$, the authors achieve the $L^{\infty}$-approximation error on a hypercube $[-M,M]^d$ using a QNN with $\mathcal{O}(\epsilon^{-2})$ parameters and $\mathcal{O}(\log_2(\epsilon^{-1}))$ qubits for any continuous function $f \in C(\mathbb{R})$ with Fourier transform satisfying
	\begin{equation}
		\label{conditionLukas}
		\int_{\mathbb{R}^d} \Vert\mathbf{\xi}\Vert_2^2 \vert\widehat{f}(\mathbf{\xi})\vert \,d \mathbf{\xi} < \infty.
	\end{equation}
	Here, in our paper, we consider $2\pi$-periodic functions. More precisely, when $d=1$ and $f$ is $k$-times continuously differentiable with a bounded $k$-th derivative, we obtain the $L^\infty$-approximation error $\epsilon > 0$ using $\mathcal{O}(\epsilon^{-1/k})$ parameters and only one qubit. Otherwise, if $d\in\mathbb{N}\cap[2,\infty)$ and $f$ has $k$-th continuous partial derivatives on all its variables, we obtain the approximation error $\epsilon > 0$ using $\mathcal{O}(\epsilon^{-(d+1)/k})$ parameters and $\mathcal{O}(\log_2(\epsilon^{-1}))$ qubit. When $k > d$, our QNN will only need to use $\mathcal{O}(\epsilon^{-1})$ parameters and $\mathcal{O}(\log_2(\epsilon^{-1}))$ qubits to obtain the approximation error $\epsilon > 0$. Hence, when restricting to periodic functions, we require asymptotically fewer parameters and the same number of qubits (as $\epsilon \rightarrow 0$) without having to impose any integrability on their Fourier transforms. 
\end{remark}

\section{Numerical Experiments}

In this section, we present an explicit algorithm for approximating both univariate and multivariate periodic functions using quantum neural networks (QNNs). Moreover, we illustrate in various numerical examples the approximation powers of QNNs.\footnote{The numerical experiments have been implemented in \texttt{Python} using the package \texttt{qiskit} on a HPC (high-performing computing) cluster of ETH Zurich. The code can be found under the following link: https://github.com/tranvietkhoa/Quantum-Neural-Network-via-Jackson-Inequality}

\subsection{Approximation of univariate periodic functions}

First, we consider the approximation of univariate periodic functions by single-qubit QNNs of the form \eqref{EqDefSingleCircuit}--\eqref{EqOutSingleCircuit}. In Algorithm~\ref{AlgUni}, we describe in pseudocode how the parameters $\theta \in \mathbb{R}^{2L+1}$ and $\phi \in \mathbb{R}^{2L+2}$ of the single-qubit QNN $U^{2L}_{\theta,\phi}$ can be computed to obtain the approximation rate \eqref{EqCorJacksonK} in Theorem~\ref{thrmJacksonK}.

\begin{small}
	\begin{algorithm}[!ht]
		\DontPrintSemicolon
		\KwInput{$K \in \mathbb{N}_0$, $N \in \mathbb{N}$, and a periodic function $f \in C_{2\pi}(\mathbb{R})$ that is $K$-times differentiable.}
		\KwOutput{QNN-based approximation $f^{2L}_{\theta,\phi}(x)$ of the form \eqref{EqOutSingleCircuit} with parameters $\theta \in \mathbb{R}^{2L+1}$ and $\phi \in \mathbb{R}^{2L+2}$ satisfying~\eqref{EqCorJacksonK}.}
		
		Define $c := (2^{K+1}-1) \Vert f \Vert_\infty \geq 0$, $r_K := \big\lceil \frac{K+3}{2} \big\rceil$, and $L := \lceil \frac{K+3}{2} \rceil \lfloor \frac{N}{2} \rfloor \in \mathbb{N}_0$.
		
		For every $n = -L,-L+1,\dots,L$ compute the Fourier coefficient $\widehat{f}(n) := \frac{1}{2\pi} \int_{-\pi}^\pi e^{-\mathbf{i} n x} f(x) dx$. 
		
		Construct the trigonometric polynomial $x \mapsto (\mathcal{T}_{N,K} f)(x) = \frac{2\pi}{\lambda_{N,K}} \sum\limits_{n=-r_K \lfloor N/2 \rfloor}^{r_K \lfloor N/2 \rfloor} m_{n,K} \widehat{f}(n) e^{\mathbf{i} n x}$ with $m_{n,K}$ defined in Proposition~\ref{LemmaJacksonK} and $\lambda_{N,K}$ defined below \eqref{EqDefKernelK}.
		
		Compute the parameters $\theta \in \mathbb{R}^{2L+1}$ and $\phi \in \mathbb{R}^{2L+2}$ via \cite[Algorithm~3]{ylwang23}, where the input is the trigonometric polynomial $x \mapsto \frac{(\mathcal{T}_{N,K} f)(x)}{c}$.
		
		Construct QNN $x \mapsto U^{2L}_{\theta,\phi}(x)$ as defined in \eqref{EqDefSingleCircuit}.
		
		\KwRet $x \mapsto f^{2L}_{\theta, \phi}(x):= \bra{0} U^{2L}_{\theta, \phi}(x)\ket{0}$.
		\caption{Learning$\:$a$\:$univariate$\:$periodic$\:$function$\:$by$\:$a$\:$quantum$\:$neural network$\:$(QNN)}
		\label{AlgUni}
	\end{algorithm}
\end{small}

Moreover, we investigate in two numerical experiments how the smoothness of the function affects the approximation error. To this end, we approximate the continuous but non-differentiable function $f_1(x) := \vert \sin(x) \vert$ and the twice but not three times differentiable function $f_{2.5}(x) := \vert \sin(x) \vert^{2.5}$ by a single-qubit QNN via $f^{2L}_{\theta,\phi}$, defined in \eqref{EqOutSingleCircuit}. According to Theorem~\ref{thrmJacksonK}, they should have approximation rates of order $\omega_{f_1}(N^{-1})$ and $N^{-2} \omega_{f''_{2.5}}(N^{-1})$, respectively. This is reflected in Figures~\ref{FigUni1}--\ref{FigUni2}, where the approximation rate of the smoother function $f_{2.5}$ decays more rapidly than $f_1$.

\begin{figure}
	\begin{minipage}[t]{0.49\textwidth}
		\centering
		\includegraphics[height=6.0cm]{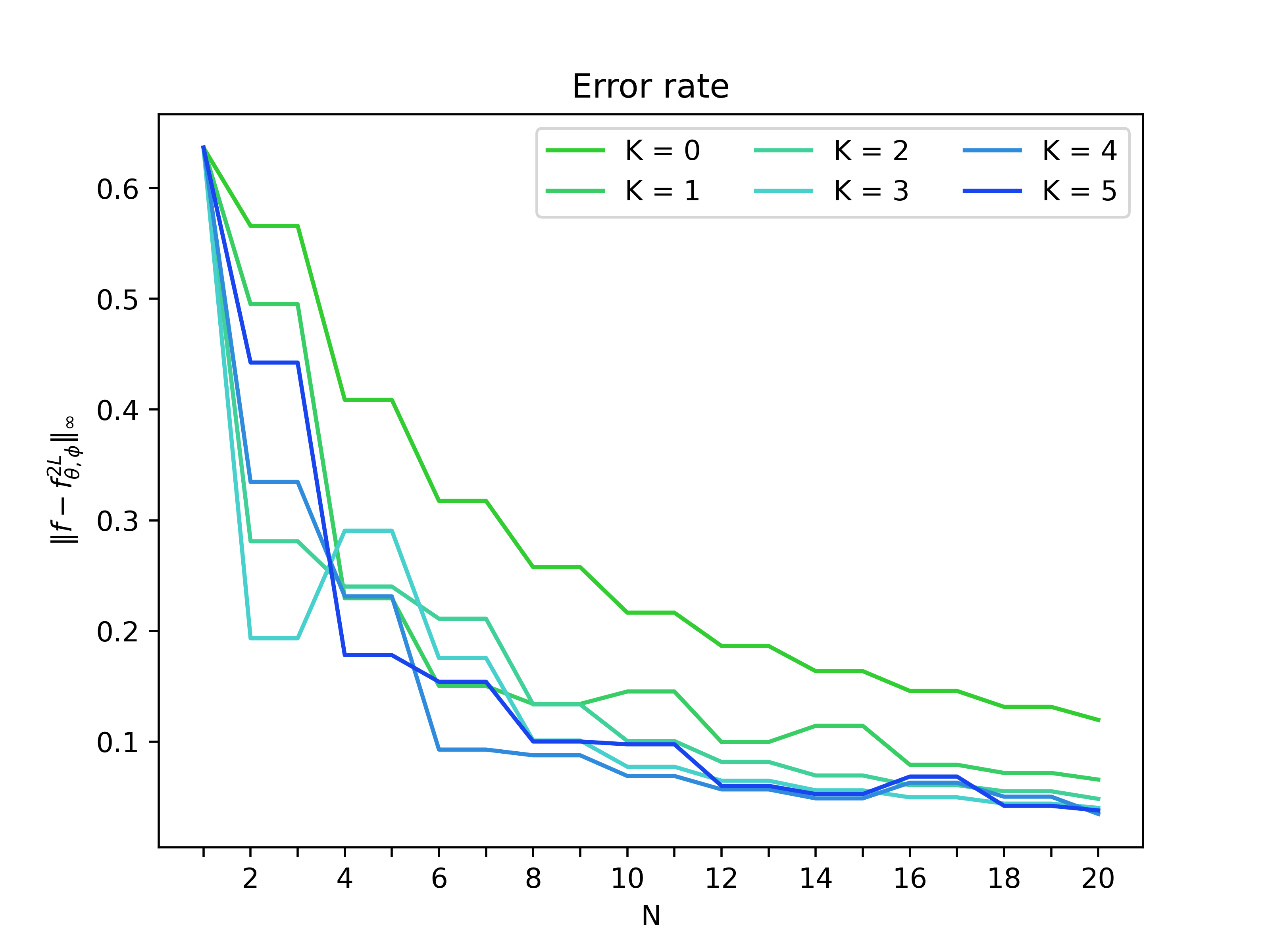}
		
		{\small {\bf (a)} Error rate}
		\vspace{0.2cm}
	\end{minipage}
	\begin{minipage}[t]{0.49\textwidth}
		\centering
		\includegraphics[height=6.0cm]{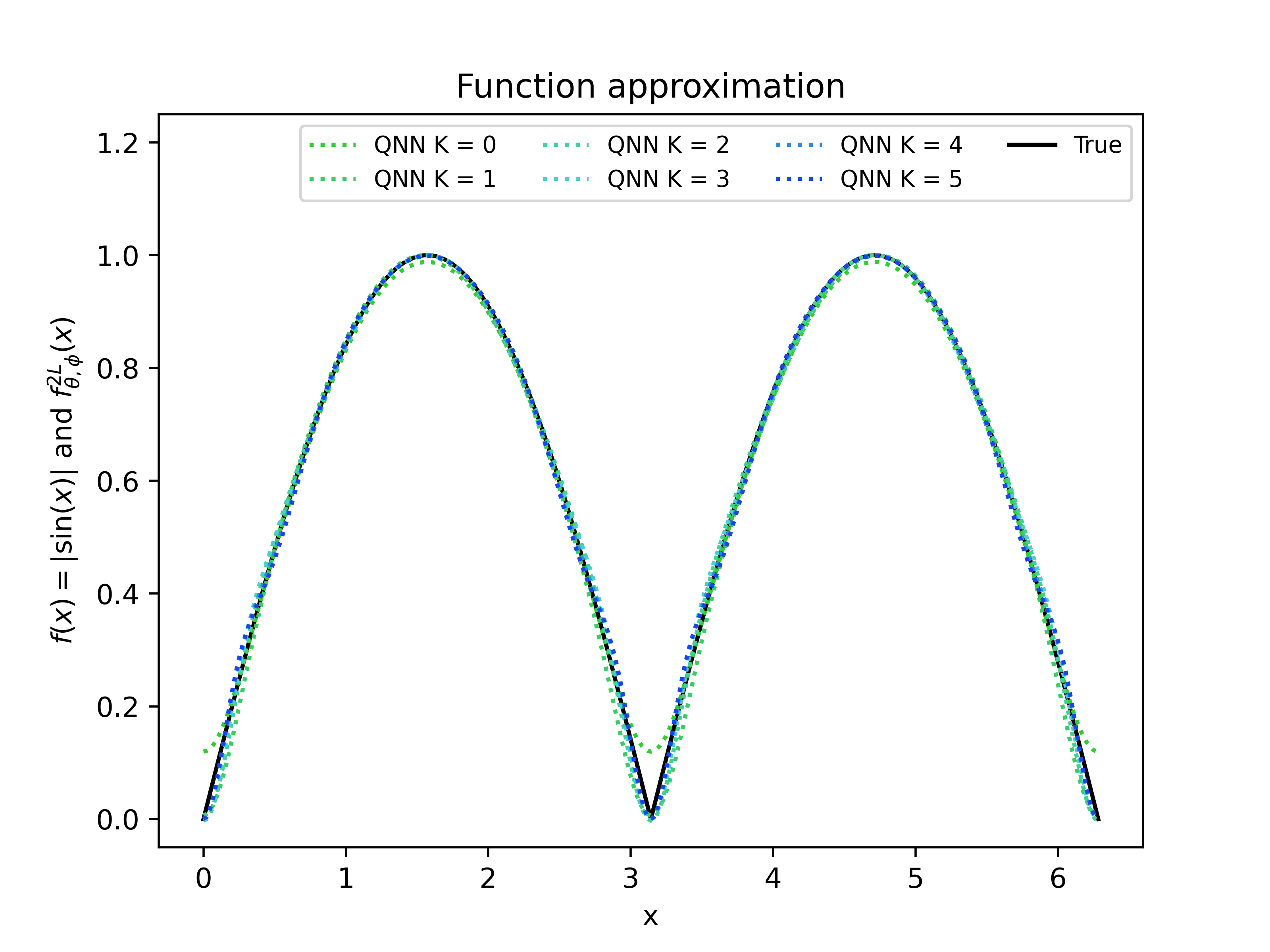}
		
		{\small {\bf (b)} Function approximation}
		\vspace{0.2cm}
	\end{minipage}
	\caption{Approximating the continuous but non-differentiable function $f_1(x) := \vert \sin(x) \vert$ by a quantum neural network (QNN) $U^{2L}_{\theta,\phi}$ with $L := \lceil \frac{K+3}{2} \rceil \lfloor \frac{N}{2} \rfloor$. In (a), the approximation error $\Vert f_1 - f^{2L}_{\theta,\phi} \Vert_\infty$ is displayed against $N \in \lbrace 1,\dots,20 \rbrace$, for $K \in \lbrace 0,\dots,5 \rbrace$. In (b), the function $f_1$ and its QNN-based approximation $f^{2L}_{\theta,\phi}$ are shown, for $N = 20$ and $K \in \lbrace 0,\dots,5 \rbrace$.}
	\label{FigUni1}
\end{figure}		

\begin{figure}
	\begin{minipage}[t]{0.49\textwidth}
		\centering
		\includegraphics[height=6.0cm]{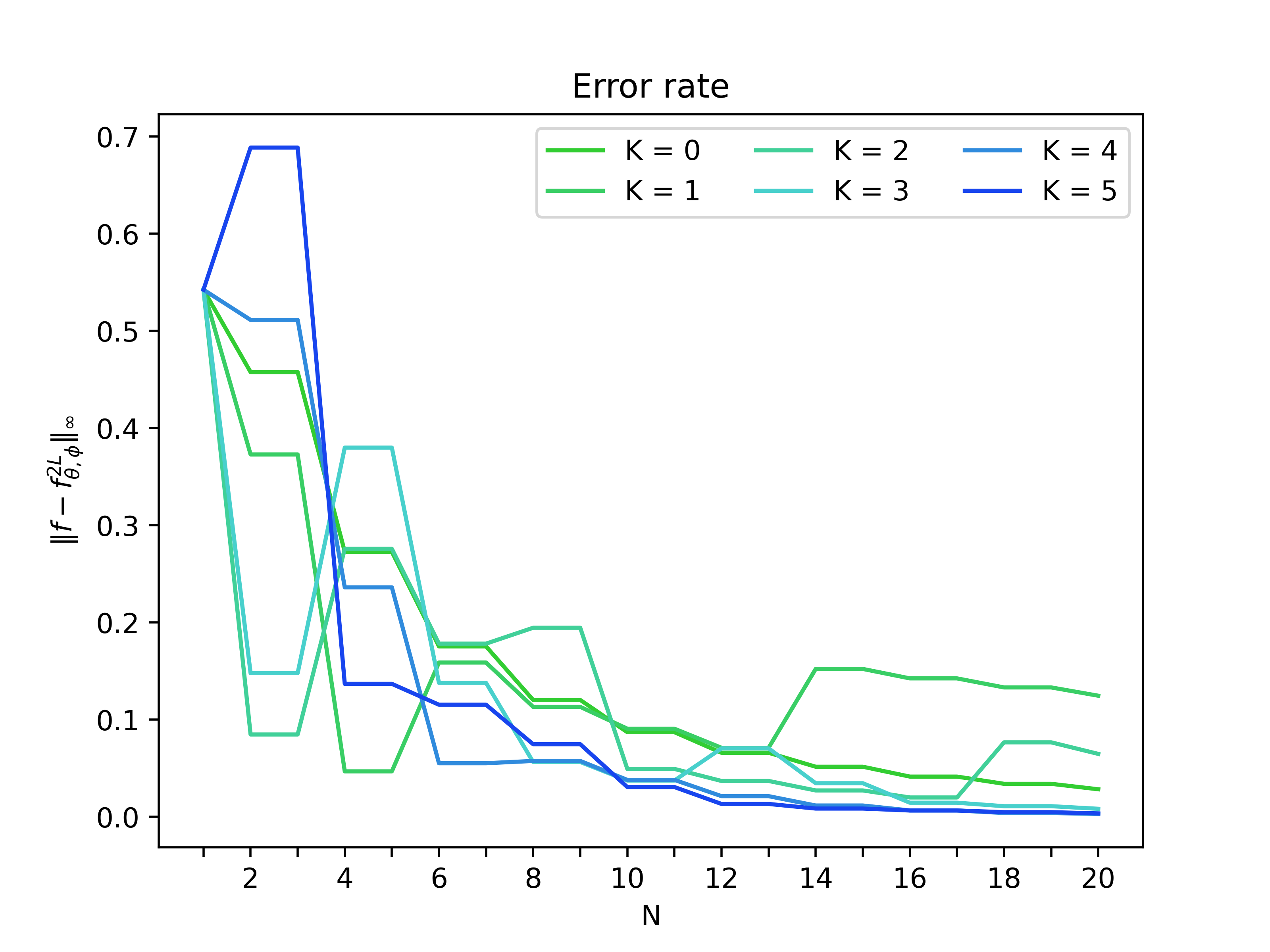}
		
		{\small {\bf (a)} Error rate}
		\vspace{0.2cm}
	\end{minipage}
	\begin{minipage}[t]{0.49\textwidth}
		\centering
		\includegraphics[height=6.0cm]{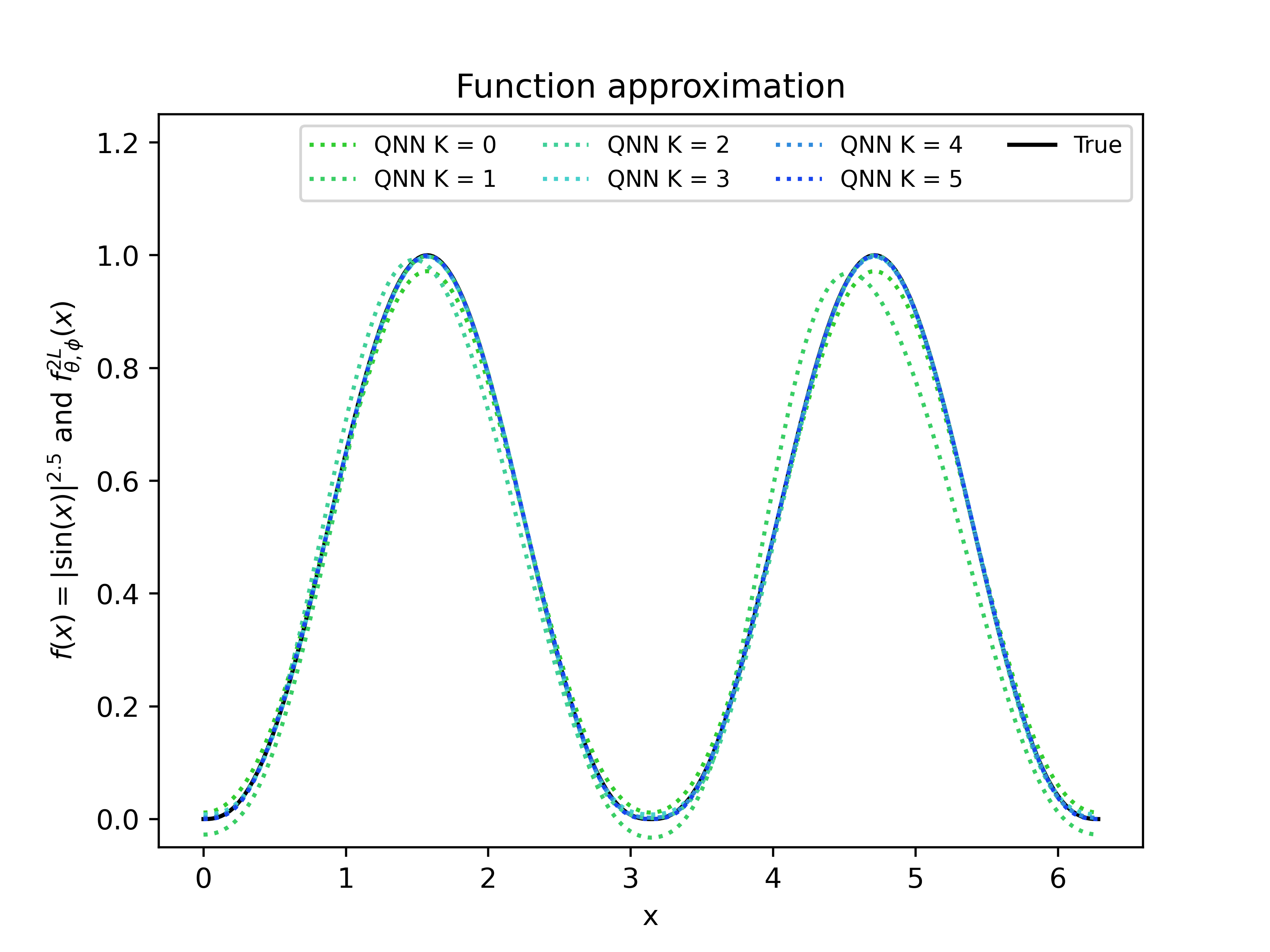}
		
		{\small {\bf (b)} Function approximation}
		\vspace{0.2cm}
	\end{minipage}
	\caption{Approximating the twice but not three times differentiable function $f_{2.5}(x) := \vert \sin(x) \vert^{2.5}$ by a quantum neural network (QNN) $U^{2L}_{\theta,\phi}$ with $L := \lceil \frac{K+3}{2} \rceil \lfloor \frac{N}{2} \rfloor$. In (a), the approximation error $\Vert f_{2.5} - f^{2L}_{\theta,\phi} \Vert_\infty$ is displayed against $N \in \lbrace 1,\dots,20 \rbrace$, for $K \in \lbrace 0,\dots,5 \rbrace$. In (b), the function $f_{2.5}$ and its QNN-based approximation $f^{2L}_{\theta,\phi}$ are shown, for $N = 20$ and $K \in \lbrace 0,\dots,5 \rbrace$.}
	\label{FigUni2}
\end{figure}

\subsection{Approximation of multivariate periodic functions}

In this section, we study the approximation of multivariate periodic function by quantum neural networks (QNNs) as defined in~\eqref{EqMultiQubitOper}--\eqref{EqMultiCircuit}. First, we describe in Algorithm~\ref{AlgMulti} using pseudocode how to find the parameters $\boldsymbol\theta := (\theta_\mathbf{n})_{-\mathbf{L}_{\mathbf{N}, \mathbf{K}} \leq \mathbf{n} \leq \mathbf{L}_{\mathbf{N}, \mathbf{K}}} \in \bigtimes_{-\mathbf{L}_{\mathbf{N}, \mathbf{K}} \leq \mathbf{n} \leq \mathbf{L}_{\mathbf{N}, \mathbf{K}}} \bigtimes_{j=1}^d \mathbb{R}^{2\vert n_j \vert+1}$ and $\boldsymbol\phi := (\phi_\mathbf{n})_{-\mathbf{L}_{\mathbf{N}, \mathbf{K}}\leq \mathbf{n} \leq \mathbf{L}_{\mathbf{N}, \mathbf{K}}} \in \bigtimes_{-\mathbf{L}_{\mathbf{N}, \mathbf{K}} \leq \mathbf{n} \leq \mathbf{L}_{\mathbf{N}, \mathbf{K}}} \bigtimes_{j=1}^d \mathbb{R}^{2\vert n_j \vert+2}$ of a $(d+q)$-qubit QNN $U^{\mathbf{L}}_{\boldsymbol\theta,\boldsymbol\phi}$ such that the approximation rate~\eqref{EqCorJacksonMultiK} in Theorem~\ref{thrmJacksonMultiK} holds true.

\begin{small}
	\begin{algorithm}[!ht]
		\DontPrintSemicolon
		\KwInput{$\mathbf{K} := (K_1, \dots,K_d) \in \mathbb{N}_0^d$, $\mathbf{N} = (N_1,\dots,N_d) \in \mathbb{N}^d$, and a periodic function $f \in C_{2\pi}(\mathbb{R}^d)$ that admits a $K_j$-th continuous partial derivative $\partial_j^{K_j} f$ into each direction $j = 1,\dots,d$.}
		\KwOutput{QNN-based approximation $f^{\mathbf{L}_{\mathbf{N}, \mathbf{K}}}_{\boldsymbol\theta,\boldsymbol\phi}$ of the form \eqref{EqMultiCircuitOut} with parameters $\boldsymbol\theta := (\theta_\mathbf{n})_{-\mathbf{L}_{\mathbf{N}, \mathbf{K}} \leq \mathbf{n} \leq \mathbf{L}_{\mathbf{N}, \mathbf{K}}}$ and $\boldsymbol\phi := (\phi_\mathbf{n})_{-\mathbf{L}_{\mathbf{N}, \mathbf{K}}\leq \mathbf{n} \leq \mathbf{L}_{\mathbf{N}, \mathbf{K}}}$ satisfying \eqref{EqCorJacksonMultiK}.}
		
		Define $c := \Vert f \Vert_\infty \prod_{j=1}^d (2^{K_j+1} - 1) \geq 0$, $\mathbf{L}_{\mathbf{N},\mathbf{K}} \in \mathbb{N}_0^d$ as in \eqref{eqL}, $\mathfrak{n} := \left|\{-\mathbf{L}_{\mathbf{N},\mathbf{K}} \leq \mathbf{n} \leq \mathbf{L}_{\mathbf{N},\mathbf{K}}\}\right|$ and $q:=\left\lceil\log_2(\mathfrak{n})\right\rceil$.
		
		For every $-\mathbf{L}_{\mathbf{N},\mathbf{K}} \leq \mathbf{n} \leq \mathbf{L}_{\mathbf{N},\mathbf{K}}$ compute the Fourier coefficient $\widehat{f}(\mathbf{n}) := \frac{1}{2\pi} \int_{-\pi}^\pi e^{-\mathbf{i} \mathbf{n}^\top \mathbf{x}} f(\mathbf{x}) d\mathbf{x}$. 
		
		Construct the trigonometric polynomial $\mathbf{x} \mapsto \frac{(\mathcal{T}_{\mathbf{N},\mathbf{K}} f)(\mathbf{x})}{c}  = \frac{(2\pi)^d}{c\left(\prod_{j=1}^d \lambda_{N_j,r_j}\right)} \sum_\mathbf{n} \mathbf{m}_{\mathbf{n},\mathbf{K}} \widehat{f}(\mathbf{n}) e^{\mathbf{i} \mathbf{n}^\top \mathbf{x}}$ with $\mathbf{m}_{\mathbf{n},\mathbf{K}}$ defined in Proposition~\ref{LemmaJacksonMultiK} and $\lambda_{N_j,r_j}$ defined below \eqref{EqDefKernelK}.
		
		\For{ $-\mathbf{L}_{\mathbf{N},\mathbf{K}} \leq \mathbf{n} \leq \mathbf{L}_{\mathbf{N},\mathbf{K}}$}{
			\For {j = 1,\dots,d}{
				\If {j = 1}{
					Compute the parameters $\theta_{\mathbf{n},1} \in \mathbb{R}^{2\vert n_1 \vert+1}$ and $\phi_{\mathbf{n},1} \in \mathbb{R}^{2\vert n_1 \vert+2}$ via \cite[Algorithm~3]{ylwang23}, where the input is the complex-valued trigonometric polynomial $x_1 \mapsto c_\mathbf{n}e^{\mathbf{i}n_1x_1}$.
					
					Construct single-qubit QNN $x_1 \mapsto U^{2\vert n_1 \vert}_{\theta_{\mathbf{n},1}, \phi_{\mathbf{n},1}}(x_1)$ as defined in \eqref{EqDefSingleCircuit}.
				}
				\Else {
					Compute the parameters $\theta_{\mathbf{n},j} \in \mathbb{R}^{2\vert n_j \vert+1}$ and $\phi_{\mathbf{n},j} \in \mathbb{R}^{2\vert n_j \vert+2}$ via \cite[Algorithm~3]{ylwang23}, where the input is the complex-valued trigonometric polynomial $x_j \mapsto e^{\mathbf{i}n_jx_j}$.
					
					Construct single-qubit QNN $x_j \mapsto U^{2\vert n_j \vert}_{\theta_{\mathbf{n},j}, \phi_{\mathbf{n},j}}(x_j)$ as defined in \eqref{EqDefSingleCircuit}.
				}
			}
			
			Construct $d$-qubit QNN $\mathbf{x} \mapsto U^{\mathbf{n}}_{\theta_{\mathbf{n}}, \phi_{\mathbf{n}}} = \bigotimes\limits_{j=1}^d U^{2\vert n_j \vert}_{\theta_{\mathbf{n},j}, \phi_{\mathbf{n},j}}(x_j)$.
			
		}
		
		Construct $\mathbf{x} \mapsto C^{\mathbf{L}_{\mathbf{N},\mathbf{K}}}_{\boldsymbol\theta,\boldsymbol\phi}(\mathbf{x}) := \sum_{i = 0}^{\mathfrak{n}-1} \ket{i} \bra{i} \otimes U^{\mathbf{n}_i}_{\theta_{\mathbf{n}_i},\phi_{\mathbf{n}_i}}(\mathbf{x}) 
		+ \sum_{i=\mathfrak{n}}^{2^q-1} \ket{i} \bra{i} \otimes X^{\otimes d}$.
		
		Construct QNN $\mathbf{x} \mapsto U^{\mathbf{L}_{\mathbf{N},\mathbf{K}}}_{\boldsymbol\theta,\boldsymbol\phi}(\mathbf{x}) := \left( H^{\otimes q} \otimes I_2^{\otimes d} \right)^\dagger C^{\mathbf{L}_{\mathbf{N},\mathbf{K}}}_{\boldsymbol\theta,\boldsymbol\phi}(\mathbf{x}) \left( H^{\otimes q} \otimes I_2^{\otimes d} \right)$.
		
		\KwRet  $\mathbf{x} \mapsto f^{\mathbf{L}_{\mathbf{N},\mathbf{K}}}_{\boldsymbol\theta,\boldsymbol\phi}(\mathbf{x}) := \bra{0}_{q+d} U^{\mathbf{L}_{\mathbf{N},\mathbf{K}}}_{\boldsymbol\theta,\boldsymbol\phi}(\mathbf{x}) \ket{0}_{q+d}$.
		\caption{Learning$\:$a$\:$multivariate$\:$periodic$\:$function$\:$by$\:$a$\:$quantum$\:$neural network$\:$(QNN)}
		\label{AlgMulti}
	\end{algorithm}
\end{small}

Moreover, we illustrate in a numerical experiment how quantum neural networks can learn a given multivariate function. To this end, we consider the \textit{heat equation}, which is a fundamental object across different disciplines such as mathematics, physics, and engineering. More precisely, for a $2\pi$-periodic and square-integrable initial condition $g: \mathbb{R}^d \rightarrow \mathbb{R}$, we consider the solution of the parabolic partial differential equation (PDE)
\begin{alignat}{2}
	\label{eqheat}
	\frac{\partial u}{\partial t}(t,\mathbf{x}) - \sum_{i=1}^d \frac{\partial^2 u}{\partial x_i^2}(t,\mathbf{x}) & = 0, \quad\quad & (t,\mathbf{x}) \in (0,T] \times \mathbb{R}^d, \\
	u(0,\mathbf{x}) & = g(\mathbf{x}), & \mathbf{x} \in \mathbb{R}^d.
\end{alignat}
Then, by following \cite[Section~2.3]{evans10}, the solution at time $t > 0$ is given by
\begin{equation*}
	u(t,\mathbf{x}) = (\Phi_t * g)(\mathbf{x}) := \int_{\mathbb{R}^d} \Phi_t(\mathbf{x}-\mathbf{y}) g(\mathbf{y}) d\mathbf{y}, \quad\quad \text{with} \quad\quad \Phi_t(\mathbf{z}) := \frac{1}{(4\pi t)^\frac{d}{2}} e^{-\frac{\Vert \mathbf{z} \Vert^2}{4t}},
\end{equation*}
which is smooth and $2\pi$-periodic. For more details, we refer to \cite[Section~2.3]{evans10}.

For fixed time $t \in (0,T]$, we now approximate in Figure~\ref{FigMulti}--\ref{FigMulti2} the solution of the heat equation $\mathbf{x} \mapsto u(t,\mathbf{x})$ with initial condition $g(\mathbf{x}) := \prod_{j=1}^d g_j(x_j)$, where
\begin{equation}
	\label{eq:heat:init_cond}
	g_j(s) := 
	\begin{cases}
		1, & \text{if } s \in \lbrace (2k\pi,(2k+1)\pi): k \in \mathbb{Z} \rbrace, \\
		0, & \text{if } s \in \lbrace 2k \pi: k \in \mathbb{Z} \rbrace, \\
		-1, & \text{if } s \in \lbrace ((2k-1)\pi,2k\pi): k \in \mathbb{Z} \rbrace.
	\end{cases}
\end{equation}
by a $(d+q)$-qubit QNN $U^{\mathbf{L}}_{\boldsymbol\theta,\boldsymbol\phi}$ via $f^{\mathbf{L}}_{\boldsymbol\theta,\boldsymbol\phi}$ as defined in \eqref{EqMultiCircuitOut}. Indeed, Figure~\ref{FigMulti}--\ref{FigMulti2} empirically demonstrate that $(d+q)$-qubit QNNs can approximate the solution of the heat equation $\mathbf{x} \mapsto u(t,\mathbf{x})$ at different times $t \in \lbrace 0.5, 1 \rbrace$. Moreover, since $\mathbf{x} \mapsto u(t,\mathbf{x})$ is a smooth function (i.e. infinitely many times continuously differentiable), the approximation rate decays more rapidly, as $K$ increases, which is in accordance with Theorem~\ref{thrmJacksonMultiK}. 

\begin{figure}
	\centering
	\begin{minipage}[t]{0.49\textwidth}
		\centering
		\includegraphics[height=6.0cm]{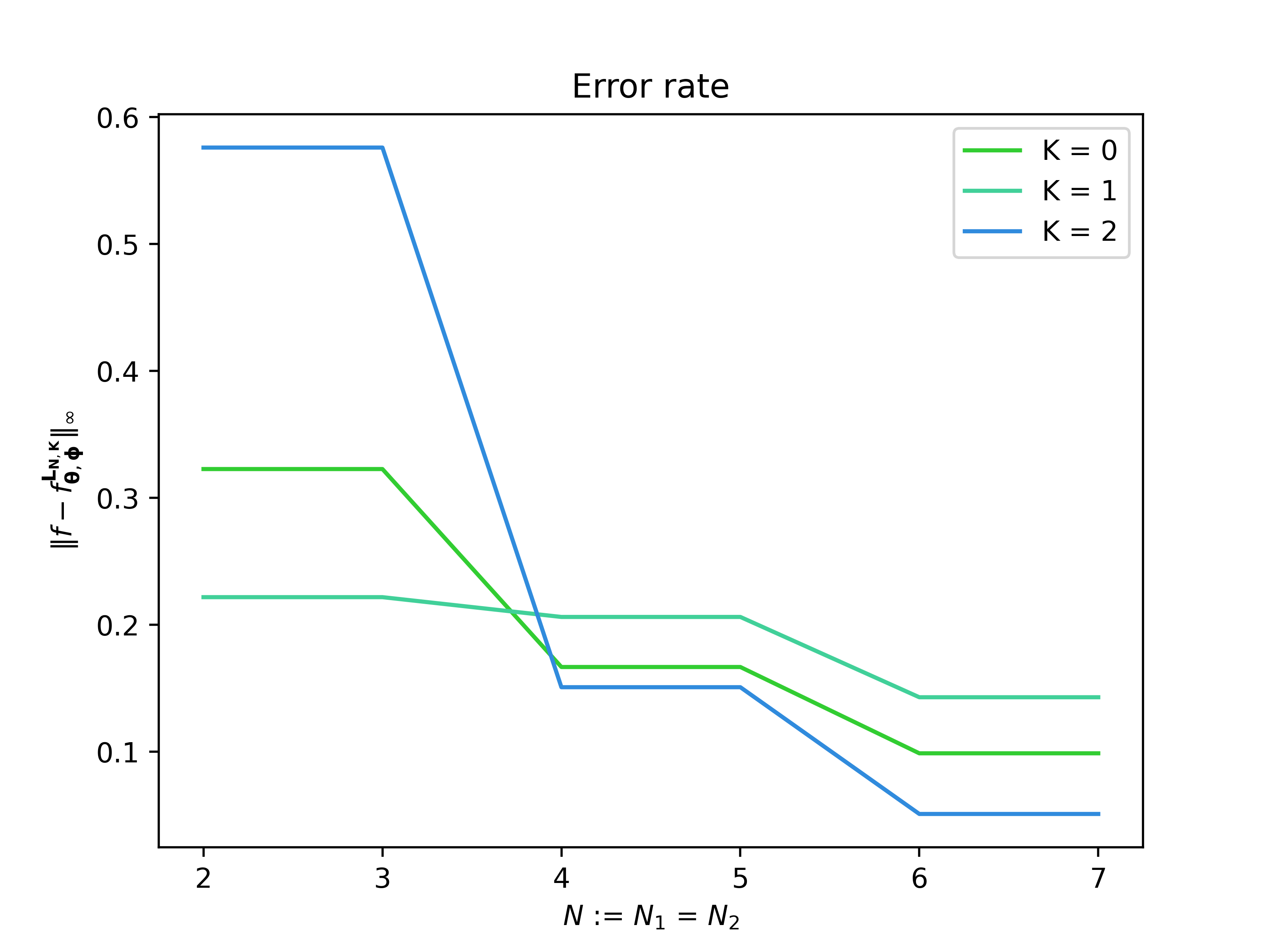}
		
		{\small {\bf (a)} Error rate}
		\vspace{0.2cm}
	\end{minipage}
	\begin{minipage}[t]{0.49\textwidth}
		\centering
		\includegraphics[height=6.0cm]{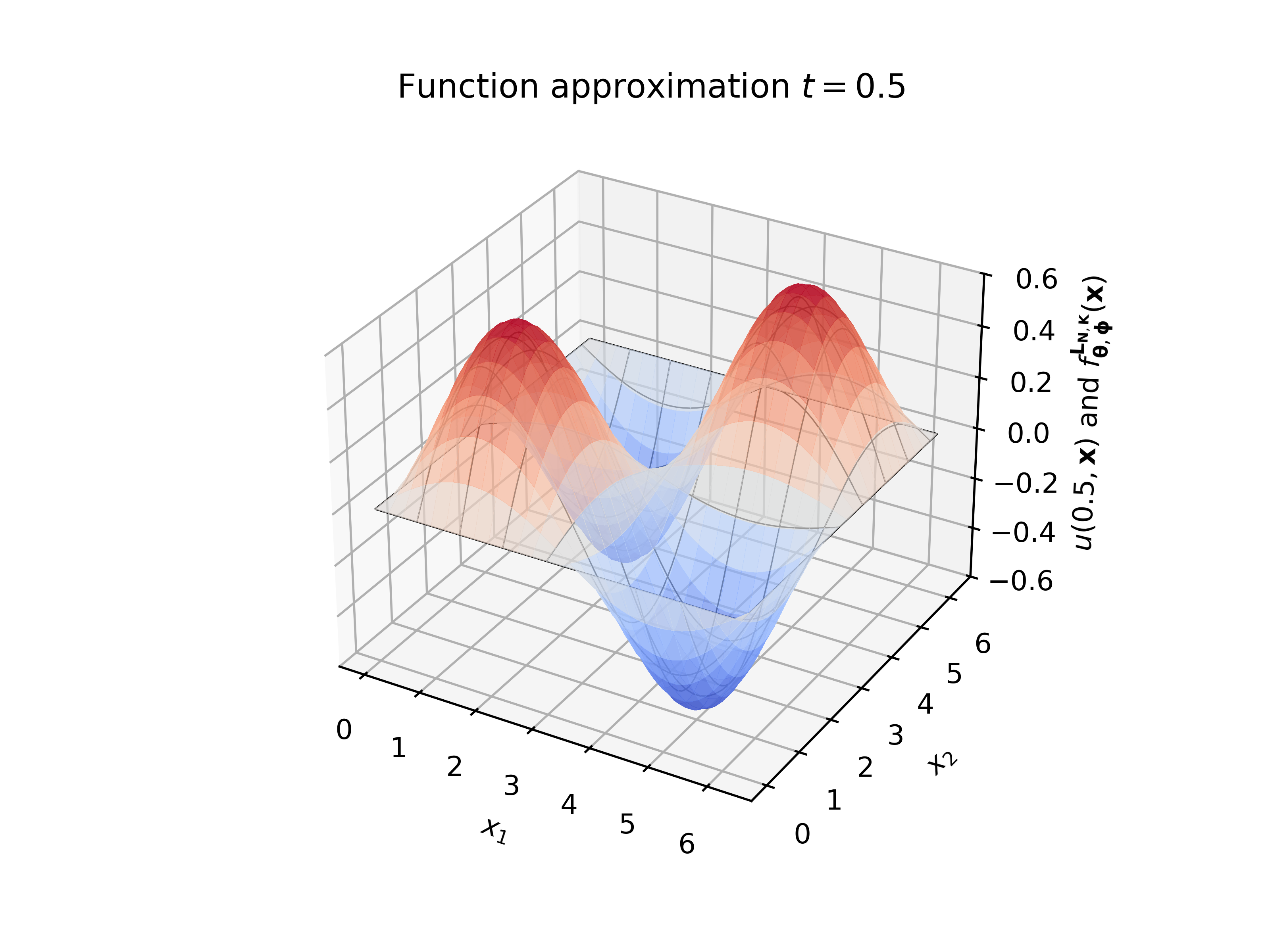}
		
		{\small {\bf (b)} Function approximation}
		\vspace{0.2cm}
	\end{minipage}
	
	\caption{Learning the solution of the heat equation $\mathbf{x} \mapsto u(0.5,\mathbf{x})$ in \eqref{eqheat} by a quantum neural network (QNN) $U^{\mathbf{L}_{\mathbf{N},\mathbf{K}}}_{\boldsymbol\theta,\boldsymbol\phi}$, with $d = 2$ and $g(\mathbf{x}) := \prod_{j=1}^d g(x_j)$, $g_j$ defined in \eqref{eq:heat:init_cond}. In (a), the approximation error $\Vert u(0.5,\mathbf{x}) - f^{\mathbf{L}_{\mathbf{N},\mathbf{K}}}_{\boldsymbol\theta,\boldsymbol\phi} \Vert_\infty$ is displayed against $N := N_1 = N_2 \in \lbrace 2,\dots,7 \rbrace$, for different $K := K_1 = K_2 \in \lbrace 0, 1, 2 \rbrace$. In (b), the function $\mathbf{x} \mapsto u(0.5,\mathbf{x})$ (wireframe) and its QNN-based approximation $\mathbf{x} \mapsto f^{\mathbf{L}_{\mathbf{N},\mathbf{K}}}_{\boldsymbol\theta,\boldsymbol\phi}(\mathbf{x})$ (colormap) are shown, for $N := N_1 = N_2 \in \lbrace 2,\dots,7 \rbrace$ and $K := K_1 = K_2 \in \lbrace 0, 1, 2 \rbrace$.}
	\label{FigMulti}
\end{figure}

\begin{figure}
	\centering
	\begin{minipage}[t]{0.49\textwidth}
		\centering
		\includegraphics[height=6.0cm]{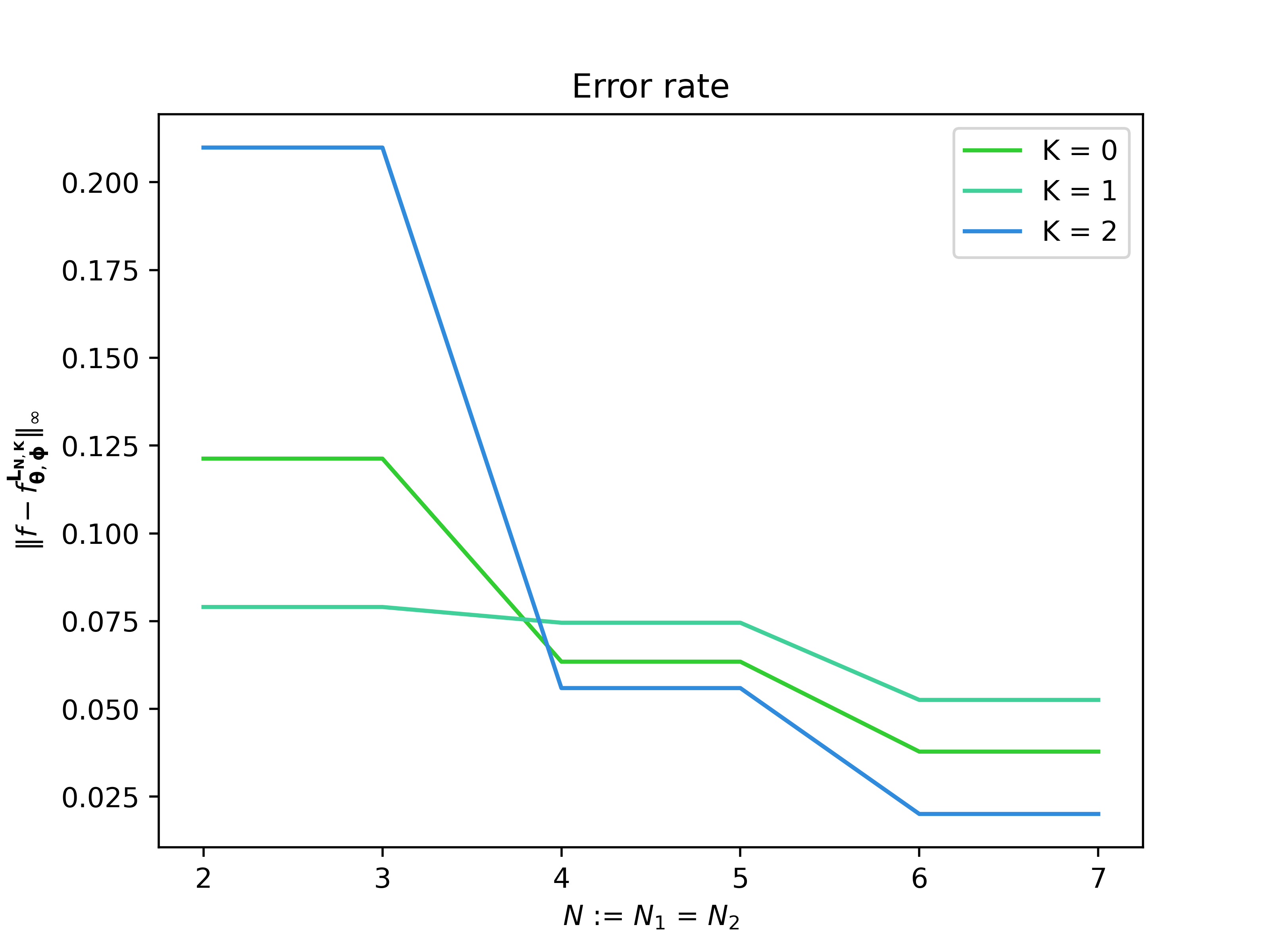}
		
		{\small {\bf (a)} Error rate}
		\vspace{0.2cm}
	\end{minipage}
	\begin{minipage}[t]{0.49\textwidth}
		\centering
		\includegraphics[height=6.0cm]{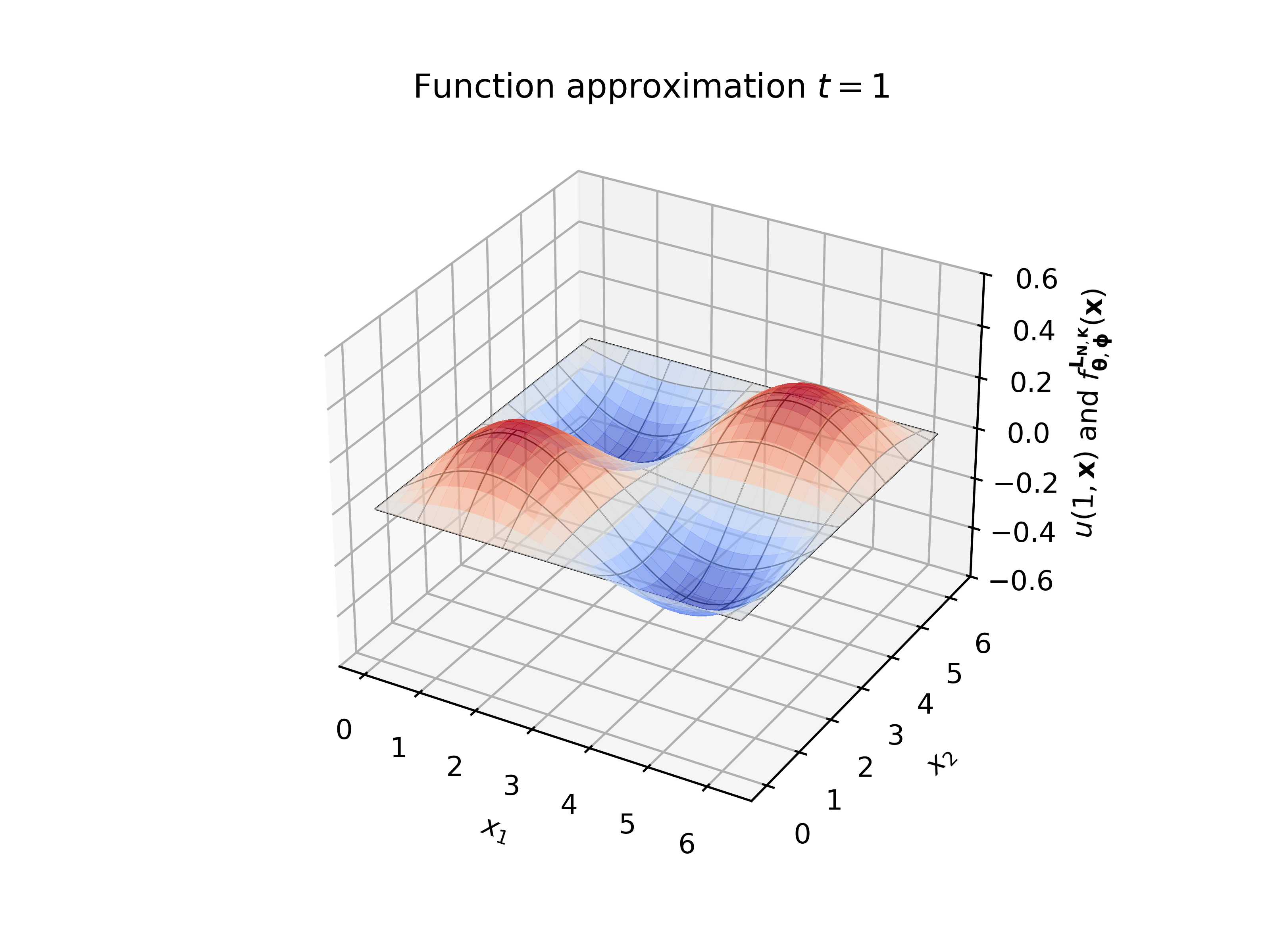}
		
		{\small {\bf (b)} Function approximation}
		\vspace{0.2cm}
	\end{minipage}
	
	\caption{Learning the solution of the heat equation $\mathbf{x} \mapsto u(1,\mathbf{x})$ in \eqref{eqheat} by a quantum neural network (QNN) $U^{\mathbf{L}_{\mathbf{N},\mathbf{K}}}_{\boldsymbol\theta,\boldsymbol\phi}$, with $d = 2$ and $g(\mathbf{x}) := \prod_{j=1}^d g(x_j)$, $g_j$ defined in \eqref{eq:heat:init_cond}. In (a), the approximation error $\Vert u(1,\mathbf{x}) - f^{\mathbf{L}_{\mathbf{N},\mathbf{K}}}_{\boldsymbol\theta,\boldsymbol\phi} \Vert_\infty$ is displayed against $N := N_1 = N_2 \in \lbrace 2,\dots,7 \rbrace$, for different $K := K_1 = K_2 \in \lbrace 0, 1, 2 \rbrace$. In (b), the function $\mathbf{x} \mapsto u(1,\mathbf{x})$ (wireframe) and its QNN-based approximation $\mathbf{x} \mapsto f^{\mathbf{L}_{\mathbf{N},\mathbf{K}}}_{\boldsymbol\theta,\boldsymbol\phi}(\mathbf{x})$ (colormap) are shown, for $N := N_1 = N_2 \in \lbrace 2,\dots,7 \rbrace$ and $K := K_1 = K_2 \in \lbrace 0, 1, 2 \rbrace$.}
	\label{FigMulti2}
\end{figure}

\section{Proof of Theorem~\ref{thrmJacksonK}}

In this section, we present the proof of Theorem~\ref{thrmJacksonK}. For a given real, 
$2\pi$-periodic univariate function $f$, the main idea is to first construct the trigonometric polynomial obtained from the Jackson inequality, which provides an approximation rate depending on the smoothness of $f$. Then, we will make use of recent results in quantum machine learning to show that trigonometric polynomials can be implemented by a suitable single-qubit quantum neural network (QNN) introduced in Section~\ref{sectionSingleQNN}.

\subsection{Univariate trigonometric polynomials and Jackson inequality}

As a cornerstone of classical approximation theory, the Jackson inequality provides an upper bound for the approximation of a continuous periodic function by a trigonometric polynomial. The inequality has been first established by David Jackson \cite{jackson11} and subsequently extended in several directions (see, e.g., \cite{achiezer13,lorentz1966,cheney82}). 

Let us first introduce complex-valued trigonometric polynomials. A function $T: \mathbb{R} \rightarrow \mathbb{C}$ is called a complex-valued \emph{(univariate) trigonometric polynomial} if there exists some $N \in \mathbb{N}_0$ and $a_n,b_n \in \mathbb{C}$, $n = 0,\dots,N$, such that for every $x \in \mathbb{R}$ it holds that
\begin{equation*}
	T(x) = a_0 + \sum_{n=1}^{N}(a_n \cos(nx) + b_n \sin(nx)).
\end{equation*}
In this case, we say that the complex-valued trigonometric polynomial is of degree $N$ if $a_N\neq 0$ or $b_N\neq 0$. Alternatively, since the functions $x \mapsto \cos(nx)$ and $x \mapsto \sin(nx)$ can be written as linear combinations of the functions $x \mapsto e^{\mathbf{i}nx}$ and $x \mapsto e^{-\mathbf{i}nx}$, we obtain the equivalent form for $T: \mathbb{R} \rightarrow \mathbb{C}$ given for every $x \in \mathbb{R}$ as
\begin{equation*}
	T(x) = \sum_{n=-N}^N c_n e^{\mathbf{i}nx},
\end{equation*}
where the coefficients $(c_n)_n$ are given by $c_0 := a_0$ as well as $c_{-n} := \frac{a_n+b_n \mathbf{i}}{2}$ and $c_n := \frac{a_n-b_n \mathbf{i}}{2}$ for $n = 1,\dots,N$.

We now define for every $K \in \mathbb{N}_0$ the number $r_K = \left\lceil \frac{K+3}{2} \right\rceil$. Moreover, we follow \cite[Section~4.3]{lorentz1966} and introduce for every fixed $N \in \mathbb{N}$, $K \in \mathbb{N}_0$ the kernel
\begin{equation}
	\label{EqDefKernelK}
	[-\pi, \pi] \ni t \quad \mapsto \quad J_{N,K}(t) := 
	\begin{cases}
		\frac{1}{\lambda_{N,K}} \left( \frac{\sin\left( \left( \left\lfloor \frac{N}{2} \right\rfloor + 1 \right) \frac{t}{2} \right)}{\sin\left( \frac{t}{2} \right)} \right)^{2r_K}, & t \neq 0, \\
		\frac{1}{\lambda_{N,K}} \left( \left\lfloor \frac{N}{2} \right\rfloor + 1 \right)^{2r_K}, & t = 0,
	\end{cases}
\end{equation}
where $\lambda_{N,K} > 0$ is a normalizing constant such that $\int_{-\pi}^\pi J_{N,K}(t) dt = 1$. Using this, we define for every fixed $K \in \mathbb{N}_0$ and $f \in C_{2\pi}(\mathbb{R})$ the function
\begin{equation}
	\label{EqDefTK}
	\mathbb{R} \ni x \quad \mapsto \quad (\mathcal{T}_{N,K} f)(x) := \int_{-\pi}^\pi J_{N,K}(t) \sum_{k=1}^{K+1} (-1)^{k+1} \binom{K+1}{k} f(x+kt) dt \in \mathbb{R}.
\end{equation}
We first show in Proposition~\ref{LemmaJacksonK} below that $\mathcal{T}_{N,K} f$ is a trigonometric polynomial. To this end, let us present the following elementary identity.

\begin{lemma}
	\label{LemmaDirichletKern}
	For every $a \in \frac{1}{2} \mathbb{Z}$ and $t \in [-\pi,\pi]$, it holds that
	\begin{equation*}
		\sum_{n=-a}^a e^{\mathbf{i} n t} = 
		\begin{cases}
			\frac{\sin\left((2a+1) \frac{t}{2} \right)}{\sin\left( \frac{t}{2} \right)}, & \text{if } t \in [-\pi,\pi] \setminus \lbrace 0 \rbrace,\\
			2a+1, & \text{if } t = 0.
		\end{cases}
	\end{equation*}
\end{lemma}
\begin{proof}
	Fix some $a \in \frac{1}{2} \mathbb{Z}$ and $t \in [-\pi,\pi]$. If $t \neq 0$, we use the formula of geometric series and that $\sin(\theta/2) = \frac{1}{2\mathbf{i}} \left( e^{\mathbf{i}\theta/2} - e^{-\mathbf{i}\theta/2} \right)$ implies $1-e^{\mathbf{i} \theta} = -2\mathbf{i} e^{\mathbf{i} \theta/2} \sin(\theta/2)$ for any $\theta \in \mathbb{R}$ to conclude that
	\begin{equation*}
		\begin{aligned}
			\sum_{n=-a}^a e^{\mathbf{i} n t} & = \sum_{k=0}^{2a} e^{\mathbf{i} (-a+k) t} = e^{-\mathbf{i} a t} \sum_{k=0}^{2a} \left( e^{\mathbf{i} t} \right)^k = e^{-\mathbf{i} a t} \frac{1-e^{\mathbf{i} (2a+1) t}}{1-e^{\mathbf{i} t}} \\
			& = e^{-\mathbf{i} a t} \frac{-2\mathbf{i} e^{\mathbf{i} (2a+1) t/2} \sin\left((2a+1) \frac{t}{2} \right)}{-2\mathbf{i} e^{\mathbf{i} t/2} \sin\left( \frac{t}{2} \right)} = \frac{\sin\left((2a+1) \frac{t}{2} \right)}{\sin\left( \frac{t}{2} \right)}.
		\end{aligned}
	\end{equation*}
	Otherwise, if $t = 0$, it holds that
	\begin{equation*}
		\sum_{n=-a}^a e^{\mathbf{i} n t} = e^{-\mathbf{i} a t} \sum_{k=0}^{2a} \underbrace{e^{\mathbf{i} k t}}_{=1} = 2a+1 = \lim_{t \rightarrow 0} \frac{\sin\left((2a+1) \frac{t}{2} \right)}{\sin\left( \frac{t}{2} \right)},
	\end{equation*}
	which completes the proof.
\end{proof}

\begin{proposition}
	\label{LemmaJacksonK}
	Let $f \in C_{2\pi}(\mathbb{R})$, let $N\in \mathbb{N}$, and $K \in \mathbb{N}_0$. Then, for every $x \in \mathbb{R}$, it holds that
	\begin{equation*}
		(\mathcal{T}_{N,K} f)(x) = \frac{2\pi}{\lambda_{N,K}} \sum_{n=-r_K \lfloor N/2 \rfloor}^{r_K \lfloor N/2 \rfloor} m_{n,K} \widehat{f}(n) e^{\mathbf{i} n x},
	\end{equation*}
	where $m_{n,K} := \sum\limits_{k \in [1,K+1] \cap \mathbb{N}: \atop k \vert n \vert \leq r_K \lfloor N/2 \rfloor} (-1)^{k+1} \binom{K+1}{k} \widetilde{m}_{k \vert n \vert}$ for $n \in \mathbb{Z}$, with $(\widetilde{m}_l)_{l \in \mathbb{Z}}$\footnote{Note that $\widetilde{m}_l$ can be calculated using the inclusion-exclusion formula.} given by $$\widetilde{m}_l := \left\vert\left\lbrace \mathbf{m} \in \lbrace -\lfloor N/2 \rfloor/2, -\lfloor N/2 \rfloor/2+1,\dots,\lfloor N/2 \rfloor/2 \rfloor \rbrace^{2r_K}: \sum\limits_{j=1}^{2r_K} \mathbf{m}_j = l \right\rbrace\right\vert$$ for $l \in \mathbb{Z}.$
\end{proposition}
\begin{proof}
	Fix some $f \in C_{2\pi}(\mathbb{R})$, let $N \in \mathbb{N}$, and $K \in \mathbb{N}_0$. Then, we first use $f(y) = \sum_{j \in \mathbb{Z}} \widehat{f}(j) e^{\mathbf{i} j y}$ for any $y \in \mathbb{R}$ and $\int_{-\pi}^\pi e^{-\mathbf{i} l t} dt = 2\pi \mathds{1}_{\lbrace 0 \rbrace}(l)$ for any $l \in \mathbb{Z}$ to conclude for every $k,n \in \mathbb{N}$ and $x \in \mathbb{R}$ that 
	\begin{equation}
		\label{EqLemmaJacksonKProof1}
		\int_{-\pi}^\pi e^{\mathbf{i} n t} f(x+kt) dt = \sum_{j \in \mathbb{Z}} \widehat{f}(j) \left( \int_{-\pi}^\pi e^{\mathbf{i} (n-jk) t} dt \right) e^{\mathbf{i} j x} =
		\begin{cases}
			2\pi \widehat{f}\left( \frac{n}{k} \right) e^{\mathbf{i} \frac{n}{k} x}, & \text{if } \frac{n}{k} \in \mathbb{Z}, \\
			0, & \text{otherwise}.
		\end{cases}
	\end{equation}
	Moreover, by inserting \eqref{EqDefKernelK}--\eqref{EqDefTK} and by using Lemma \ref{LemmaDirichletKern}, we obtain for every $x \in \mathbb{R}$ that
	\begin{equation*}
		\begin{aligned}
			(\mathcal{T}_{N,K} f)(x) & = \frac{1}{\lambda_{N,K}} \int_{-\pi}^\pi \left( \frac{\sin\left( \left( \left\lfloor \frac{N}{2} \right\rfloor + 1 \right) \frac{t}{2} \right)}{\sin\left( \frac{t}{2} \right)} \right)^{2r_K} \sum_{k=1}^{K+1} (-1)^{k+1} \binom{K+1}{k} f(x+kt) dt \\
			& = \frac{1}{\lambda_{N,K}} \sum_{k=1}^{K+1} (-1)^{k+1} \binom{K+1}{k} \int_{-\pi}^\pi \left( \sum_{n=-\lfloor N/2 \rfloor/2}^{\lfloor N/2 \rfloor/2} e^{\mathbf{i} n t} \right)^{2r_K} f(x+kt) dt.
		\end{aligned}
	\end{equation*}
	Hence, by using the multinomial theorem (with coefficients $(\widetilde{m}_n)_{n\in\mathbb{Z}}$), the identity \eqref{EqLemmaJacksonKProof1}, and reordering the sum $\widetilde{n} \mapsto n/k \in \mathbb{Z}$, it follows for every $x \in \mathbb{R}$ that
	\begin{equation*}
		\begin{aligned}
			(\mathcal{T}_{N,K} f)(x) & = \frac{1}{\lambda_{N,K}} \sum_{k=1}^{K+1} (-1)^{k+1} \binom{K+1}{k} \sum_{n=-r_K \lfloor N/2 \rfloor}^{r_K \lfloor N/2 \rfloor} \widetilde{m}_n \int_{-\pi}^\pi e^{\mathbf{i} n t} f(x+kt) dt \\
			& = \frac{2\pi}{\lambda_{N,K}} \sum_{n=-r_K \lfloor N/2 \rfloor}^{r_K \lfloor N/2 \rfloor} \sum_{k=1}^{K+1} (-1)^{k+1} \binom{K+1}{k} \widetilde{m}_n \mathds{1}_\mathbb{Z}\left( \frac{n}{k} \right) \widehat{f}\left( \frac{n}{k} \right) e^{\mathbf{i} \frac{n}{k} x} \\ 
			& = \frac{2\pi}{\lambda_{N,K}} \sum_{\widetilde{n}=-r_K \lfloor N/2 \rfloor}^{r_K \lfloor N/2 \rfloor} \sum_{k=1 \atop k \vert \widetilde{n} \vert \leq r \lfloor N/2 \rfloor}^{K+1} (-1)^{k+1} \binom{K+1}{k} \widetilde{m}_{k\widetilde{n}} \widehat{f}(\widetilde{n}) e^{\mathbf{i} \widetilde{n} x}, \\
			& = \frac{2\pi}{\lambda_{N,K}} \sum_{n=-r_K \lfloor N/2 \rfloor}^{r_K \lfloor N/2 \rfloor} m_{n,K} \widehat{f}(n) e^{\mathbf{i} n x},
		\end{aligned}
	\end{equation*}
	which completes the proof.
\end{proof}

Now, we recall the Jackson inequality for $K$-times continuously differentiable periodic functions.

\begin{proposition}[{\!\!\cite[Theorem~4.3, p.~57]{lorentz1966}}]
	\label{propJacksonK}
	For $K \in \mathbb{N}_0$, let $f \in C_{2\pi}(\mathbb{R})$ be $K$-times continuously differentiable. Then, there exists a constant $C_K > 0$ (independent of $f$) such that for every $N \in \mathbb{N}$ it holds that
	\begin{equation}
		\label{EqThmJacksonK}
		\sup_{x \in \mathbb{R}} \left\vert f(x) - (\mathcal{T}_{N,K} f)(x) \right\vert \leq \frac{C_K \omega_{f^{(k)}}\left( \frac{1}{N} \right)}{N^K}.
	\end{equation}
	Moreover, if $f \in C_{2\pi}(\mathbb{R})$ is $(K+1)$-times continuously differentiable, then the right-hand side of \eqref{EqThmJacksonK} can be upper bounded by $\frac{C_K \Vert f^{(K+1)} \Vert_\infty}{N^{K+1}}$, where $f^{(K+1)}$ denotes the $(K+1)^{\text{th}}$ derivative of~$f$.
\end{proposition}

\subsection{Error bound by quantum neural network}

In this section, we follow the ideas of \cite{yu22,ylwang23} and recall that every bounded complex-valued trigonometric polynomial can be represented as the output of the single-qubit QNN of the form \eqref{EqDefSingleCircuit}.

\begin{proposition}[{\!\!\cite[Corollary~2]{ylwang23}}]
	\label{PropSingleQubit}
	For $L \in \mathbb{N}_0$, let $T: \mathbb{R} \rightarrow \mathbb{C}$ be a complex-valued trigonometric polynomial of degree $L$ satisfying $|T(x)| \leq 1$ for all $x \in \mathbb{R}$. Then, there exist parameters $\theta = (\theta_0, \theta_1, \dots, \theta_{2L}) \in \mathbb{R}^{2L+1}$ and $\phi = (\varphi, \phi_0, \phi_1, \dots, \phi_{2L}) \in \mathbb{R}^{2L+2}$ such that for every $x \in \mathbb{R}$ it holds that
	\begin{equation*}
		f^{2L}_{\theta,\phi}(x) := \Bra{0} U^{2L}_{\theta,\phi}(x) \Ket{0} = T(x).
	\end{equation*}
\end{proposition}

Now, we are in the position to prove our first main result (Theorem \ref{thrmJacksonK}), establishing the approximation rates for single-qubit quantum neural networks.

\begin{proof}[Proof of Theorem \ref{thrmJacksonK}]
	Fix some $N \in \mathbb{N}$ and define the function $g \in C_{2\pi}(\mathbb{R})$ by $g(x) := \frac{f(x)}{c}$ for $x \in \mathbb{R}$, which is also $K$-times continuously differentiable. Then, by using that $J_{N,K}(t) \geq 0$ for all $t \in [-\pi,\pi]$ and that $\int_{-\pi}^\pi J_{N,K}(t) dt = 1$, we have
	\begin{equation}
		\label{EqCorJacksonKProof1}
		\begin{aligned}
			\left\vert (\mathcal{T}_{N,K} g)(x) \right\vert & \leq \int_{-\pi}^\pi \vert J_{N,K}(t) \vert \sum_{k=1}^{K+1} \binom{K+1}{k} \vert g(x+kt) \vert dt \\
			& \leq \Vert g \Vert_\infty \sum_{k=1}^{K+1} \binom{K+1}{k} \int_{-\pi}^\pi J_{N,K}(t) dt \\
			& \leq (2^{K+1}-1) \Vert g \Vert_\infty = (2^{K+1}-1) \frac{\Vert f \Vert_\infty}{c} \leq 1.
		\end{aligned}
	\end{equation}
	Hence, by using Proposition~\ref{PropSingleQubit} (applied to $\mathcal{T}_{N,K} g$ having degree $r_K \lfloor N/2 \rfloor$), there exist parameters $\theta \in \mathbb{R}^{2 r_K \lfloor N/2 \rfloor+1}$ and $\phi \in \mathbb{R}^{2 r_K \lfloor N/2 \rfloor+2}$ such that for every $x \in \mathbb{R}$ it holds that
	\begin{equation*}
		(\mathcal{T}_{N,K} g)(x) = \Bra{0} U^{2r_K \lfloor N/2 \rfloor}_{\theta,\phi}(x) \Ket{0}.
	\end{equation*}
	Finally, by using this and Proposition~\ref{propJacksonK} (applied to $g \in C_{2\pi}(\mathbb{R})$) and that $c \cdot \omega_{g^{(k)}} = \omega_{f^{(k)}}$, it follows that
	\begin{equation*}
		\begin{aligned}
			\sup_{x \in \mathbb{R}} \left\vert f(x) - c \cdot  f^{2r_K \lfloor N/2 \rfloor}_{\theta,\phi}(x) \right\vert & = c \sup_{x \in \mathbb{R}} \left\vert g(x) - (\mathcal{T}_{N,K} g)(x) \right\vert \\
			& \leq c \frac{C_K \omega_{g^{(k)}}\left( \frac{1}{N} \right)}{N^K} \\
			& = \frac{C_K \omega_{f^{(k)}}\left( \frac{1}{N} \right)}{N^K},
		\end{aligned}
	\end{equation*}
	which completes the proof of the first part. The second part follows from the second part of Proposition \ref{propJacksonK}.
\end{proof}

\section{Proof of Theorem~\ref{thrmJacksonMultiK}}

In this section, we present the proof of Theorem \ref{thrmJacksonMultiK}. To this end, we use the multivariate generalization of the Jackson inequality (see, e.g., \cite{lorentz1966,cheney82}), which relies on an iterative application of the univariate Jackson inequality through all variables. In order to implement the multivariate trigonometric polynomial of the multivariate Jackson inequality, we use the idea of \emph{(parametrized) quantum circuits (PQCs)} in \cite{yu25}, which in turn is based on the \emph{linear combination of unitaries (LCU)} in \cite{childs12}.

\subsection{Multivariate trigonometric polynomials and Jackson inequality}

Let us first introduce complex-valued multivariate trigonometric polynomials. A function $T:\mathbb{R}^d \rightarrow \mathbb{C}$ is called a \emph{$d$-variate complex-valued trigonometric polynomial} if there exist some $\mathbf{N} := (N_1,\dots,N_d) \in \mathbb{N}_0^d$ and $a_{\mathbf{n}}, b_{\mathbf{n}} \in \mathbb{C}$, $\mathbf{n} := (n_1,\dots,n_d) \in \mathbb{N}_0^d$ and $-\mathbf{N} \leq \mathbf{n} \leq \mathbf{N}$, such that for every $\mathbf{x} \in \mathbb{R}^d$ it holds that
\begin{equation*}
	T(\mathbf{x}) = \sum_{-\mathbf{N} \leq \mathbf{n} \leq \mathbf{N}} \left( a_{\mathbf{n}} \cos\left( \mathbf{n}^\top \mathbf{x} \right) + b_{\mathbf{n}} \sin\left( \mathbf{n}^\top \mathbf{x} \right) \right).
\end{equation*}
In this case, we say that the complex-valued trigonometric polynomial $T:\mathbb{R}^d \rightarrow \mathbb{C}$ has degree $\mathbf{N} \in \mathbb{N}_0^d$ if $a_{\mathbf{N}} \neq 0$ or $b_{\mathbf{N}} \neq 0$. Alternatively, since the functions $\mathbf{x} \mapsto \cos\left( \mathbf{n}^\top \mathbf{x} \right)$ and $\mathbf{x} \mapsto \sin\left( \mathbf{n}^\top \mathbf{x} \right)$ can be written as linear combinations of the functions $\mathbf{x} \mapsto e^{\mathbf{i} \mathbf{n}^\top \mathbf{x}}$ and $\mathbf{x} \mapsto e^{-\mathbf{i} \mathbf{n}^\top \mathbf{x}}$, we obtain the equivalent form of the $d$-variate complex-valued trigonometric polynomial $T: \mathbb{R}^d \rightarrow \mathbb{C}$ given for every $\mathbf{x} \in \mathbb{R}^d$ as
\begin{equation*}
	T(\mathbf{x}) = \sum_{-\mathbf{N} \leq \mathbf{n} \leq\mathbf{N}} c_{\mathbf{n}} e^{\mathbf{i} \mathbf{n}^\top \mathbf{x}},
\end{equation*}
where the coefficients $(c_{\mathbf{n}})_{\mathbf{n}}$ are given by $c_\mathbf{0} := a_\mathbf{0}$ as well as $c_{-\mathbf{n}} := \frac{1}{2} \left( a_{\mathbf{n}}+b_{\mathbf{n}} \mathbf{i} \right)$ and $c_{\mathbf{n}} := \frac{1}{2} \left( a_{\mathbf{n}}-b_{\mathbf{n}} \mathbf{i} \right)$ for $-\mathbf{N} \leq \mathbf{n} \leq \mathbf{N}$, $\mathbf{n} \neq \lbrace 0 \rbrace^d$. 

We first recall the Jackson inequality for multivariate functions. To this end, we follow \cite[Section~6.3]{lorentz1966} and introduce for every function $f \in C_{2\pi}(\mathbb{R}^d)$ as well as for every multi-indices $\mathbf{K} := (K_1,\dots,K_d) \in \mathbb{N}_0^d$ and $\mathbf{N} := (N_1,\dots,N_d) \in \mathbb{N}^d$ the recursively defined function $\mathcal{I}_j(\mathbf{x})$, $\mathbf{x} \in \mathbb{R}^d$ such that
\begin{equation*}
	\mathcal{I}_j(\mathbf{x}) := 
	\begin{cases}
		\left(\mathcal{T}^{x_j}_{N_j, K_j} \left(\mathcal{I}_{j-1}\right) \right) (\mathbf{x}), & j = 1,2,\dots,d, \\
		f(\mathbf{x}), & j = 0,
	\end{cases}
\end{equation*}
and define $\left(\mathcal{T}_{\mathbf{N}, \mathbf{K}} \right)f(\mathbf{x}) := \mathcal{I}_d(\mathbf{x})$, where at each iteration $j = 1,\dots,d$ the function $\mathcal{T}^{x_j}_{N_j,K_j} g: \mathbb{R}^d \rightarrow \mathbb{R}$ with $g \in C_{2\pi}(\mathbb{R}^d)$ denotes the application of the operator $\mathcal{T}_{N_j,K_j}$ (defined in \eqref{EqDefTK}) to the function $\mathbb{R} \ni x_j \mapsto g(x_1,\dots,x_{j-1},x_j,x_{j+1},\dots,x_d) \in \mathbb{R}$. Then, by iteratively applying Proposition~\ref{LemmaJacksonK}, we conclude that $\mathcal{T}_{\mathbf{N},\mathbf{K}} f$ is a multivariate trigonometric polynomial.

\begin{proposition}
	\label{LemmaJacksonMultiK}
	Let $f \in C_{2\pi}(\mathbb{R}^d)$, let $\mathbf{K} := (K_1,\dots,K_d) \in \mathbb{N}_0^d$ and $\mathbf{N} := (N_1,\dots,N_d) \in \mathbb{N}^d$. Then, for every $\mathbf{x} \in \mathbb{R}^d$, it holds that
	\begin{equation*}
		(\mathcal{T}_{\mathbf{N},\mathbf{K}} f)(\mathbf{x}) = \frac{(2\pi)^d}{\prod_{j=1}^d \lambda_{N_j,r_j}} \sum_\mathbf{n} \mathbf{m}_{\mathbf{n},\mathbf{K}} \widehat{f}(\mathbf{n}) e^{\mathbf{i} \mathbf{n}^\top \mathbf{x}},
	\end{equation*}
	where the sum is taken over all $\mathbf{n} := (n_1,\dots,n_d) \in \mathbb{Z}^d$ with $\vert n_j \vert \leq r_j \lfloor N_j/2 \rfloor$ for all $j = 1,\dots,d$, and where $\mathbf{m}_{\mathbf{n},\mathbf{K}} := \prod_{j=1}^d m_{n_j,K_j}$ with $m_{n_j,K_j}$ defined in Proposition~\ref{LemmaJacksonK}.
\end{proposition}
\begin{proof}
	Fix some $\mathbf{K} := (K_1,\dots,K_d) \in \mathbb{N}_0^d$ and $\mathbf{N} := (N_1,\dots,N_d) \in \mathbb{N}^d$. Then, by using for every $j = 1,\dots,d$ the notation $g_j := \mathcal{T}^{x_j}_{N_j,K_j} \big( \mathcal{T}^{x_{j-1}}_{N_{j-1},K_{j-1}} ( \dots (\mathcal{T}^{x_1}_{N_1,K_1} f) \dots ) \Big)$ and Proposition~\ref{LemmaJacksonK}, it follows for every $\mathbf{x} \in \mathbb{R}^d$ that
	\begin{equation}
		\label{EqLemmaJacksonMultiKProof1}
		\begin{aligned}
			g_j(\mathbf{x}) & = (\mathcal{T}^{x_j}_{N_j,K_j} g_{j-1})(\mathbf{x}) \\
			& = \frac{2\pi}{\lambda_{N_j,K_j}} \sum_{n_j=-r_{K_j} \lfloor N_j/2 \rfloor}^{r_{K_j} \lfloor N_j/2 \rfloor} m_{n_j,K_j} \left( \frac{1}{2\pi} \int_{-\pi}^\pi e^{-\mathbf{i} n_j s_j} g_{j-1}(x_1,\dots,x_{j-1},s_j,x_{j+1},\dots,x_d) ds_j \right) e^{\mathbf{i} n_j x_j}.
		\end{aligned}
	\end{equation}
	Hence, by iteratively inserting \eqref{EqLemmaJacksonMultiKProof1} and using linearity, we conclude for every $\mathbf{x} \in \mathbb{R}^d$ that
	\begin{equation*}
		\begin{aligned}
			& (\mathcal{T}_{\mathbf{N},\mathbf{K}} f)(\mathbf{x}) = (\mathcal{T}^{x_d}_{N_d,K_d} g_{d-1})(\mathbf{x}) \\
			& \quad\quad = \frac{2\pi}{\lambda_{N_d,K_d}} \sum_{n_d=-r_{K_d} \lfloor N_d/2 \rfloor}^{r_{K_d} \lfloor N_d/2 \rfloor} m_{n_d,K_d} \left( \frac{1}{2\pi} \int_{-\pi}^\pi e^{-\mathbf{i} n_d s_d} g_{d-1}(x_1,\dots,x_{d-1},s_d) ds_d \right) e^{\mathbf{i} n_d x_d} \\
			& \quad\quad = \frac{(2\pi)^2}{\lambda_{N_d,K_d} \lambda_{N_{d-1},K_{d-1}}} \sum_{n_d=-r_{K_d} \lfloor N_d/2 \rfloor}^{r_{K_d} \lfloor N_d/2 \rfloor} \sum_{n_{d-1}=-r_{K_{d-1}} \lfloor N_{d-1}/2 \rfloor}^{r_{K_{d-1}} \lfloor N_{d-1}/2 \rfloor} m_{n_d,K_d} m_{n_{d-1},K_{d-1}} \\
			& \quad\quad\quad\quad \cdot \left( \frac{1}{(2\pi)^2} \int_{-\pi}^\pi \int_{-\pi}^\pi e^{-\mathbf{i} n_{d-1} s_{d-1}} e^{-\mathbf{i} n_d s_d} g_{d-2}(x_1,\dots,x_{d-2},s_{d-1},s_d) ds_{d-1} ds_d \right) e^{\mathbf{i} n_{d-1} x_{d-1}} e^{\mathbf{i} n_d x_d} \\
			& \quad\quad = \dots = \frac{(2\pi)^d}{\lambda_{N_d,K_d} \cdots \lambda_{N_1,K_1}} \sum_{n_d=-r_{K_d} \lfloor N_d/2 \rfloor}^{r_{K_d} \lfloor N_d/2 \rfloor} \cdots \sum_{n_1=-r_{K_1} \lfloor N_1/2 \rfloor}^{r_{K_1} \lfloor N_1/2 \rfloor} m_{n_d,K_d} \cdots m_{n_1,K_1} \\
			& \quad\quad\quad\quad\quad\quad \cdot \left( \frac{1}{(2\pi)^d} \int_{-\pi}^\pi \cdots \int_{-\pi}^\pi e^{-\mathbf{i} n_1 s_1} \cdots e^{-\mathbf{i} n_d s_d} f(s_1,\dots,s_d) ds_1 \cdots ds_d \right) e^{\mathbf{i} n_1 x_1} \cdots e^{\mathbf{i} n_d x_d} \\
			& \quad\quad = \frac{(2\pi)^d}{\prod_{j=1}^d \lambda_{N_j,K_j}} \sum_\mathbf{n} \mathbf{m}_{\mathbf{n},\mathbf{K}} \widehat{f}(\mathbf{n}) e^{\mathbf{i} \mathbf{n}^\top \mathbf{x}},
		\end{aligned}
	\end{equation*}
	which completes the proof.
\end{proof}

Now, we recall the Jackson inequality for multivariate functions, which establishes an upper bound for approximating any given function $f \in C_{2\pi}(\mathbb{R}^d)$ by the trigonometric polynomial $\mathcal{T}_{\mathbf{N},\mathbf{K}} f$.

\pagebreak

\begin{proposition}[{\!\!\cite[Theorem~6.6, p.~87]{lorentz1966}}]
	\label{propJacksonMultiK}
	For $\mathbf{K} := (K_1,\dots,K_d) \in \mathbb{N}_0^d$ and $j=1,\dots,d$, let $f \in C_{2\pi}(\mathbb{R}^d)$ have a $K_j$-th continuous partial derivative $\partial_j^{K_j} f := \frac{\partial^{K_j} f}{\partial x_j^{K_j}} f: \mathbb{R}^d \rightarrow \mathbb{R}$. Then, there exists a constant $C_{\mathbf{K}} > 0$ (independent of $f$) such that for every $\mathbf{N} := (N_1,\dots,N_d) \in \mathbb{N}^d$ it holds that
	\begin{equation}
		\label{EqThmJacksonMultiK}
		\sup_{x \in \mathbb{R}^d} \left\vert f(\mathbf{x}) - (\mathcal{T}_{\mathbf{N},\mathbf{K}} f)(\mathbf{x}) \right\vert \leq C_\mathbf{K} \sum_{j=1}^d \frac{\omega_{\partial_j^{K_j} f}\big( \frac{1}{N_j} \big)}{N_j^{K_j}}.
	\end{equation}
	Moreover, if $f \in C_{2\pi}(\mathbb{R}^d)$ has a ($K_j+1$)-th continuous partial derivative $\partial_j^{K_j+1} f: \mathbb{R}^d \rightarrow \mathbb{R}$, $j = 1,\dots,d$, then the right-hand side of \eqref{EqThmJacksonMultiK} can be upper bounded by $C_\mathbf{K} \sum_{j=1}^d \frac{\Vert \partial_j^{K_j+1} f \Vert_\infty}{N_j^{K_j+1}}$.
\end{proposition}

\subsection{Error bound by quantum neural network}

We now represent every bounded multivariate trigonometric polynomial by a multi-qubit QNN. 

\begin{proposition}
	\label{proposition-multivariate}
	For $\mathbf{L} := (L_1,\dots,L_d) \in \mathbb{N}_0^d$, let $T: \mathbb{R}^d \rightarrow \mathbb{C}$ be a complex-valued trigonometric polynomial of degree $\mathbf{L} \in \mathbb{N}_0^d$ satisfying $\vert T(\mathbf{x})\vert \leq 1$ for all $\mathbf{x} \in \mathbb{R}^d$, and define $q := \lceil \log_2(\mathfrak{n}) \rceil$ with $\mathfrak{n} := | \{\mathbf{n} \in \mathbb{Z}^d: -\mathbf{L} \leq \mathbf{n} \leq \mathbf{L}\}|$. Then, there exists a multi-qubit QNN as defined in \eqref{EqMultiCircuit} with parameters $\boldsymbol\theta := (\theta_\mathbf{n})_{-\mathbf{L} \leq \mathbf{n} \leq \mathbf{L}} \in \bigtimes_{-\mathbf{L} \leq \mathbf{n} \leq \mathbf{L}} \bigtimes_{j=1}^d \mathbb{R}^{2\vert n_j \vert+1}$ and $\boldsymbol\phi := (\phi_\mathbf{n})_{-\mathbf{L} \leq \mathbf{n} \leq \mathbf{L}} \in \bigtimes_{-\mathbf{L} \leq \mathbf{n} \leq \mathbf{L}} \bigtimes_{j=1}^d \mathbb{R}^{2\vert n_j \vert+2}$ such that for every $\mathbf{x} \in \mathbb{R}^d$ it holds that
	\begin{equation*}
		\Bra{\mathbf{0}}_{q+d} U^{\mathbf{L}}_{\boldsymbol\theta,\boldsymbol\phi}(\mathbf{x}) \Ket{\mathbf{0}}_{q+d} = 2^{-q} \, T(\mathbf{x}).
	\end{equation*}
\end{proposition}
\begin{proof}
	For some fixed $\mathbf{L} \in \mathbb{N}_0^d$, let $T: \mathbb{R}^d \rightarrow \mathbb{C}$ be a complex-valued trigonometric polynomial of the form $T(x) = \sum\limits_{-\mathbf{L} \leq \mathbf{n} \leq \mathbf{L}} c_{\mathbf{n}} e^{\mathbf{i} \mathbf{n}^\top \mathbf{x}}$, $\mathbf{x} \in \mathbb{R}^d$, which satisfies $\vert T(\mathbf{x}) \vert \leq 1$ for all $\mathbf{x} \in \mathbb{R}^d$. Then, for every fixed $\mathbf{n} \in \mathbb{Z}^d$ with $-\mathbf{L} \leq \mathbf{n} \leq \mathbf{L}$, we observe that for every $\mathbf{x} := (x_1,\dots,x_d)^\top \in \mathbb{R}^d$ it holds that $\left\vert c_{\mathbf{n}} e^{\mathbf{i} n_1 x_1} \right\vert = \vert c_{\mathbf{n}} \vert = \big\vert \widehat{T}(\mathbf{n}) \big\vert = \big\vert (2\pi)^{-d} \int_{[-\pi,\pi]^d} e^{\mathbf{i} \mathbf{n}^\top \mathbf{s}} T(\mathbf{s}) d\mathbf{s} \big\vert \leq \sup_{\mathbf{s} \in [-\pi,\pi]^d} \vert T(\mathbf{s}) \vert \leq 1$ and that $\left\vert e^{\mathbf{i} n_j x_j} \right\vert = 1$, $j = 2,\dots,d$. Hence, we can apply $d$-times Proposition~\ref{PropSingleQubit} to obtain some parameters $\theta_{\mathbf{n},j} \in \mathbb{R}^{2\vert n_j \vert+1}$ and $\phi_{\mathbf{n},j} \in \mathbb{R}^{2\vert n_j \vert+2}$, $j = 1,\dots,d$, such that for every $\mathbf{x} := (x_1,\dots,x_d)^\top \in \mathbb{R}^d$ it holds that
	\begin{equation*}
		\begin{aligned}
			c_{\mathbf{n}} e^{\mathbf{i} n_1 x_1} & = \Bra{0} U^{2\vert n_1 \vert}_{\theta_{\mathbf{n},1},\phi_{\mathbf{n},1}}(x_1) \Ket{0}, \\
			e^{\mathbf{i} n_2 x_2} & = \Bra{0} U^{2\vert n_2 \vert}_{\theta_{\mathbf{n},2},\phi_{\mathbf{n},2}}(x_2) \Ket{0}, \\
			& \,\,\, \vdots \\
			e^{\mathbf{i} n_d x_d} & = \Bra{0} U^{2\vert n_d \vert}_{\theta_{\mathbf{n},d},\phi_{\mathbf{n},d}}(x_d) \Ket{0}.
		\end{aligned}
	\end{equation*}
	Therefore, by defining the $d$-qubit operator $U^{\mathbf{n}}_{\theta_\mathbf{n},\phi_\mathbf{n}}(\mathbf{x}) := \bigotimes_{j=1}^d U^{2\vert n_j \vert}_{\theta_{\mathbf{n},j},\phi_{\mathbf{n},j}}(x_j)$ with parameters $\theta_\mathbf{n} := (\theta_{\mathbf{n},1},\dots,\theta_{\mathbf{n},d}) \in \bigtimes_{j=1}^d \mathbb{R}^{2\vert n_j \vert+1}$ and $\phi_\mathbf{n} := (\phi_{\mathbf{n},1},\dots,\phi_{\mathbf{n},d}) \in \bigtimes_{j=1}^d \mathbb{R}^{2\vert n_j \vert+2}$, it follows for every $\mathbf{x} := (x_1,\dots,x_d)^\top \in \mathbb{R}^d$ that
	\begin{equation*}
		\Bra{\mathbf{0}}_d U^{\mathbf{n}}_{\theta_\mathbf{n},\phi_\mathbf{n}}(\mathbf{x}) \Ket{\mathbf{0}}_d = \prod_{j=1}^d \Bra{0} U^{2\vert n_j \vert}_{\theta_{\mathbf{n},j},\phi_{\mathbf{n},j}}(x_j) \Ket{0} = c_\mathbf{n} \prod_{j=1}^d e^{\mathbf{i} n_j x_j} = c_{\mathbf{n}} e^{\mathbf{i} \mathbf{n}^\top \mathbf{x}}.
	\end{equation*}
	Thus, by using the definition of the ($q+d$)-qubit operator linearity $U^{\mathbf{L}}_{\boldsymbol\theta,\boldsymbol\phi}(\mathbf{x})$ in \eqref{EqMultiCircuit}, the ordering $\{\mathbf{n}_0,\dots,\mathbf{n}_{\mathfrak{n}-1}\}$ of $\{\mathbf{n} \in \mathbb{Z}^d, -\mathbf{L} \leq \mathbf{n} \leq \mathbf{L}\}$, and that $\Bra{\mathbf{0}}_d X^{\otimes d} \Ket{\mathbf{0}}_d = \prod_{j=1}^d \langle 0 \vert X \vert 0 \rangle = 0$, we conclude for every $\mathbf{x} := (x_1,\dots,x_d)^\top \in \mathbb{R}^d$ that
	\[
	\begin{aligned}
		& \Bra{\mathbf{0}}_{q+d} U^{\mathbf{L}}_{\boldsymbol\theta,\boldsymbol\phi}(\mathbf{x}) \Ket{\mathbf{0}}_{q+d} = \Bra{(\psi_q,\mathbf{0})}_{q+d} C^{\mathbf{L}}_{\boldsymbol\theta,\boldsymbol\phi}(\mathbf{x}) \Ket{(\psi_q,\mathbf{0})}_{q+d} \\
		& \quad\quad = \sum_{i = 0}^{\mathfrak{n}-1} \underbrace{\langle \psi_q| i\rangle}_{=1/\sqrt{2^q}} \Bra{\mathbf{0}}_d U^{\mathbf{n}}_{\theta_\mathbf{n},\phi_\mathbf{n}}(\mathbf{x}) \Ket{\mathbf{0}}_d \underbrace{\langle i| \psi_q \rangle}_{=1/\sqrt{2^q}} + \sum_{i=\mathfrak{n}}^{2^q-1} \underbrace{\langle \psi_q| i \rangle}_{=1/\sqrt{2^q}} \Bra{\mathbf{0}}_d X^{\otimes d} \Ket{\mathbf{0}}_d \underbrace{\langle i| \psi_q \rangle}_{=1/\sqrt{2^q}} \\
		& \quad\quad = 2^{-q} \sum_{-\mathbf{L} \leq \mathbf{n} \leq\mathbf{L}} c_{\mathbf{n}} e^{\mathbf{i} \mathbf{n}^\top \mathbf{x}},
	\end{aligned}
	\]
	which completes the proof.
\end{proof}

We are now in the position to prove our second main result (Theorem \ref{thrmJacksonMultiK}), establishing the approximation rates for multi-qubit quantum neural networks.

\begin{proof}[Proof of Theorem \ref{thrmJacksonMultiK}]
	Fix some $\mathbf{N} := (N_1,\dots,N_d) \in \mathbb{N}^d$ and define $g \in C_{2\pi}(\mathbb{R}^d)$ by $g(\mathbf{x}) := \frac{f(\mathbf{x})}{c}$ for $\mathbf{x} \in \mathbb{R}^d$, which also has a $K_j$-th continuous partial derivative $\partial_j^{K_j} g: \mathbb{R}^d \rightarrow \mathbb{R}$,  $j = 1,\dots,d$. Then, by iteratively applying \eqref{EqCorJacksonKProof1} to $\mathbb{R} \ni x_j \mapsto \mathcal{T}^{x_j}_{N_j,K_j} \Big( \mathcal{T}^{x_{j-1}}_{N_{j-1},K_{j-1}} \big( \dots (\mathcal{T}^{x_1}_{N_1,K_1} g) \dots \big) \Big)(x_1,\dots,x_j,\dots,x_d) \in \mathbb{R}$ for any fixed $(x_1,\dots,x_{j-1},x_{j+1},\dots,x_d) \in \mathbb{R}^{d-1}$, $j = 1,\dots,d$, we conclude for every $\mathbf{x} \in \mathbb{R}^d$ that
	\begin{equation*}
		\begin{aligned}
			\left\vert (\mathcal{T}_{\mathbf{N},\mathbf{K}} g)(\mathbf{x}) \right\vert & = \left\vert \bigg( \mathcal{T}^{x_d}_{N_d,K_d} \Big( \mathcal{T}^{x_{d-1}}_{N_{d-1},K_{d-1}} \big( \dots (\mathcal{T}^{x_1}_{N_1,K_1} g) \dots \big) \Big) \bigg)(\mathbf{x}) \right\vert \\
			& \leq (2^{K_d+1}-1) \sup_{x_d \in \mathbb{R}} \left\vert \Big( \mathcal{T}^{x_{d-1}}_{N_{d-1},K_{d-1}} \big( \dots (\mathcal{T}^{x_1}_{N_1,K_1} g) \dots \big) \Big)(x_1,\dots,x_d) \right\vert \\
			& \leq \dots \leq (2^{K_d+1}-1) \cdots (2^{K_2+1} -1)\sup_{(x_2,\dots,x_d) \in \mathbb{R}^{d-1}} \left\vert (\mathcal{T}^{x_1}_{N_1,K_1} g)(x_1,\dots,x_d) \right\vert \\
			& \leq \left( \prod\limits_{j=1}^d(2^{K_j+1}-1) \right)\sup_{(x_1,\dots,x_d) \in \mathbb{R}^d} \left\vert g(x_1,\dots,x_d) \right\vert \\
			& = \left(\prod\limits_{j=1}^d(2^{K_j+1}-1) \right) \Vert g \Vert_\infty = \left(\prod\limits_{j=1}^d(2^{K_j+1}-1)\right) \frac{\Vert f \Vert_\infty}{c} \leq 1.
		\end{aligned}
	\end{equation*}
	Hence, by applying Proposition~\ref{proposition-multivariate} (applied to $\mathcal{T}_{\mathbf{N},\mathbf{K}} g$ having degree $(r_{K_j} \lfloor N_j/2 \rfloor)_{j=1,\dots,d}$), there exist some parameters $\boldsymbol\theta := (\theta_\mathbf{n})_{-\mathbf{L}_{\mathbf{N},\mathbf{K}} \leq \mathbf{n} \leq \mathbf{L}_{\mathbf{N},\mathbf{K}}} \in \bigtimes_{-\mathbf{L}_{\mathbf{N},\mathbf{K}} \leq \mathbf{n} \leq \mathbf{L}_{\mathbf{N},\mathbf{K}}} \bigtimes_{j=1}^d \mathbb{R}^{2\vert n_j \vert+1}$ and $\boldsymbol\phi := (\phi_\mathbf{n})_{-\mathbf{L}_{\mathbf{N},\mathbf{K}} \leq \mathbf{n} \leq \mathbf{L}_{\mathbf{N},\mathbf{K}}} \in \bigtimes_{-\mathbf{L}_{\mathbf{N},\mathbf{K}} \leq \mathbf{n} \leq \mathbf{L}_{\mathbf{N},\mathbf{K}}} \bigtimes_{j=1}^d \mathbb{R}^{2\vert n_j \vert+2}$ such that for every $\mathbf{x} \in \mathbb{R}^d$ it holds that
	\begin{equation*}
		(\mathcal{T}_{\mathbf{N},\mathbf{K}} g)(\mathbf{x}) = 2^q \cdot \Bra{\mathbf{0}}_{q+d} U^{\mathbf{L}_{\mathbf{N},\mathbf{K}}}_{\boldsymbol\theta,\boldsymbol\phi}(\mathbf{x}) \Ket{\mathbf{0}}_{q+d} := 2^q \cdot f^{\mathbf{L}_{\mathbf{N},\mathbf{K}}}_{\boldsymbol\theta,\boldsymbol\phi} (\mathbf{x}).
	\end{equation*}
	Finally, by using this, Proposition~\ref{propJacksonMultiK} (applied to $g \in C_{2\pi}(\mathbb{R}^d)$), and that $c\cdot \omega_{\partial_j^{K_j} g} = \omega_{\partial_j^{K_j} f}$, it follows that
	\begin{equation*}
		\begin{aligned}
			\sup_{\mathbf{x} \in \mathbb{R}^d} \left\vert f(\mathbf{x}) - c \cdot 2^q \cdot f^{\mathbf{L}_{\mathbf{N},\mathbf{K}}}_{\boldsymbol\theta,\boldsymbol\phi} (\mathbf{x}) \right\vert & = c \sup_{\mathbf{x} \in \mathbb{R}^d} \left\vert g(\mathbf{x}) - (\mathcal{T}_{\mathbf{N},\mathbf{K}} g)(\mathbf{x}) \right\vert \\
			& \leq c C_\mathbf{K} \sum_{j=1}^d \frac{\omega_{\partial_j^{K_j} g}\left( \frac{1}{N} \right)}{N_j^{K_j}} \\
			& = C_\mathbf{K} \sum_{j=1}^d \frac{\omega_{\partial_j^{K_j} f}\left( \frac{1}{N} \right)}{N_j^{K_j}},
		\end{aligned}
	\end{equation*}
	which completes the proof of the first part of the theorem. The second part follows directly from the second part of Proposition \ref{propJacksonMultiK}.
\end{proof}

\bibliographystyle{abbrv}
\bibliography{mybib}

\end{document}